\documentclass[11pt]{article}
\usepackage[utf8]{inputenc}
\usepackage[T1]{fontenc}
\usepackage{note}
\allowdisplaybreaks
\makeatletter
\newcommand{\TODO}[1]{\textcolor{red}{[TODO\@ifnotempty{#1}{: #1}]}}
\newcommand{\sam}[1]{\textcolor{blue}{[sam\@ifnotempty{#1}{: #1}]}}
\newcommand{\fred}[1]{\textcolor{purple}{[fred\@ifnotempty{#1}{: #1}]}}
\makeatother

\usepackage[ruled,vlined, linesnumbered]{algorithm2e}
	\SetKwFor{While}{While}{}{}
	\SetKwFor{For}{For}{}{}
	\SetKwIF{If}{ElseIf}{Else}{If}{}{elif}{Else}{}
	\SetKw{Return}{Return}

 \title{Robust and Heavy-Tailed  Mean Estimation Made Simple, via Regret Minimization }
\date{} 
\author{Samuel B. Hopkins\thanks{
Department of Electrical Engineering and Computer Sciences, UC Berkeley. Email: \texttt{\{hopkins, z0\}@berkeley.edu}} \and Jerry Li \thanks{Microsoft Research AI. Email: \texttt{jerrl@microsoft.com}}\and Fred Zhang\footnotemark[1]}

\begin{document}

\maketitle
\thispagestyle{empty}

\begin{abstract}
     We study the problem of estimating the mean of a distribution in high dimensions when either the samples are  adversarially corrupted or the distribution is heavy-tailed. Recent developments in robust statistics have established   efficient and (near) optimal procedures for both settings.  However, the algorithms developed on each side tend to be sophisticated and do not directly transfer to the other, with many of them having ad-hoc or complicated analyses.
     
     In this paper, we provide a meta-problem and   a duality theorem that lead to a new unified view on  robust and heavy-tailed mean estimation in high dimensions. We show that the meta-problem can be solved either by   a variant of the   \textsc{Filter} algorithm from the recent literature on robust estimation or by the quantum entropy scoring scheme (QUE), due to Dong, Hopkins  and Li (NeurIPS '19). By leveraging  our duality theorem, these results translate into  simple and efficient algorithms for both robust and heavy-tailed settings.  Furthermore, the QUE-based procedure has run-time that matches the fastest known algorithms on both fronts.
     
     Our analysis of \textsc{Filter} is through the classic regret bound of the  multiplicative weights update method. This connection allows us to avoid the technical complications in previous works and  improve upon the run-time analysis of a gradient-descent-based algorithm for robust mean estimation by Cheng, Diakonikolas, Ge and Soltanolkotabi (ICML '20).
\end{abstract}
\newpage
\setcounter{page}{1}

\section{Introduction}
Learning from high-dimensional data in the presence of outliers is a central task in modern statistics and machine learning.
Outliers have many sources.
Modern data sets can be exposed to random corruptions or even malicious tampering, as in data poison attacks.
Data drawn from heavy-tailed distributions can naturally contain outlying samples---heavy-tailed data are found often in network science, biology, and beyond \cite{faloutsos1999power,leskovec2005graphs,barabasi2005origin,albert2005scale}.
Minimizing the effect of outliers on the performance of learning algorithms is therefore a key challenge for statistics and computer science.

\emph{Robust statistics}---that is, statistics in the presence of outliers---has been studied formally since at least the 1960s, and informally since long before \cite{huber1964robust, tukey60}.
However, handling outliers in high dimensions presents significant computational challenges.
Classical robust estimators (such as the Tukey median) suffer from worst-case computational hardness, while na\"ive computationally-efficient algorithms (e.g., throwing out atypical-looking samples) have far-from-optimal rates of error.
In the last five years, however, numerous works have developed sophisticated, efficient algorithms with optimal error rates for a variety of problems in high-dimensional robust statistics.
Despite significant recent progress, many basic algorithmic questions remain unanswered, and many algorithms and rigorous approaches to analyzing them remain complex and \emph{ad hoc}.

In this work, we revisit the most fundamental high-dimensional estimation problem, estimating the mean of a distribution from samples, in the following two basic and widely-studied robust settings.
In each case, $X_1,\ldots,X_n \in \mathbb{R}^d$ are independent samples from an unknown $d$-dimensional distribution $D$ with mean $\mu \in \mathbb{R}^d$ and (finite) covariance $\Sigma \in \mathbb{R}^{d \times d}$.
\begin{itemize}
    \item \emph{Robust mean estimation:} Given $Y_1,\ldots,Y_n \in \mathbb{R}^d$ such that $Y_i = X_i$ except for $\epsilon n$ choices of $i$, estimate the mean $\mu$.
    We interpret the $\epsilon n$ \emph{contaminated} samples $Y_i \neq X_i$ as corruptions introduced by a malicious adversary.
    Na\"ive estimators such as the empirical mean can suffer arbitrarily-high inaccuracy as a result of these malicious samples.
    
    \item \emph{Heavy-tailed mean estimation:} Given $X_1,\ldots,X_n$, estimate $\mu$ by an estimator $\hat{\mu}$ such that $\|\mu - \hat{\mu}\|$ is small with high probability (or equivalently, estimate $\mu$ with optimal confidence intervals).
    Since our only assumption about $D$ is that it has finite covariance, $D$ may have heavy tails.
    Standard estimators such as the empirical mean can therefore be poorly concentrated.
\end{itemize}

A significant amount of recent work in statistics and computer science has led to an array of algorithms for both problems with provably-optimal rates of error and increasingly-fast running times, both in theory and experiments \cite{lai2016agnostic,diakonikolas2019robust,dong2019quantum, cheng2019high,hopkins2018mean,cherapanamjeri2019fast,lecue2019robust, lei2019fast}.
However, several questions remain, which we address in this work.

First, the relationship between heavy-tailed and robust mean estimation is still murky: while algorithms are known which simultaneously solve both problems to information-theoretic optimality \cite{lecue2019robust}, we lack general conditions under which algorithms for one problem also solve the other.
This suggests:
\begin{quote}
   \emph{Question 1: Is there a formal connection between robust mean estimation and heavy-tailed mean estimation which can be exploited by efficient algorithms?}
\end{quote}

Second, iterated \emph{sample downweighting} (or pruning) is arguably the most natural approach to statistics with outliers---indeed, the \emph{filter}, one of the first computationally efficient algorithms for optimal robust mean estimation \cite{diakonikolas2019robust}) takes this approach---but rigorous analyses of filter-style algorithms remain \emph{ad hoc}.
Other iterative methods, such as gradient descent, suffer the same fate: they are simple-to-describe algorithms which require significant creativity to analyze \cite{cheng2020high}.
We ask:
\begin{quote}
  \emph{Question 2: Is there a simple and principled approach to rigorously analyze iterative algorithms for robust and heavy-tailed mean estimation?}
\end{quote}

\subsection{Our Results}
Our main contribution in this work is a simple and unified treatment of iterative methods for robust and heavy-tailed mean estimation.

Addressing Question 1, we begin by distilling a simple meta-problem, which we call \emph{spectral sample reweighing}.
While several variants of spectral sample reweighing are implicit in recent algorithmic robust statistics literature, our work is the first to separate the problem from the context of robust mean estimation and show the reduction from heavy-tailed mean estimation.
The goal in spectral sample reweighing is to take a dataset $\{x_i\}_{i \in [n]} \subseteq \mathbb{R}^d$, reweigh the vectors $x_i$ according to some weights $w_i \in [0,1]$, and find a center $\nu \in \mathbb{R}^d$ such that after reweighing the maximum eigenvalue of the covariance $\sum_{i \leq n} w_i (x_i - \nu)(x_i - \nu)^\top$ is as small as possible. 

\begin{definition}[$(\alpha,\epsilon)$ spectral sample reweighing, informal, see ~\autoref{def:ssr}]
For $\epsilon \in (0,1/2)$, let $\mathcal{W}_{n,\epsilon} = \{w \in \Delta_n \, : \, \|w\|_\infty \leq \tfrac 1 {(1-\epsilon)n} \}$ be the set of probability distributions on $[n]$ with bounded $\ell_\infty$ norm.
Let $\alpha \geq 1$.
Given $\{x_i\}_{i=1}^n$ in $\mathbb{R}^d$, an $\alpha$-approximate spectral sample reweighing algorithm returns a probability distribution $w \in \mathcal{W}_{n, 3 \epsilon}$ and a \emph{spectral center} $\nu \in \mathbb{R}^d$ such that
\[
  \left \| \sum_{i \leq n} w_i (x_i - \nu)(x_i - \nu)^\top \right \| \leq \alpha \cdot \min_{w' \in \mathcal{W}_{n,\epsilon}, \nu' \in \mathbb{R}^d} \left \| \sum_{i \leq n} w_i' (x_i - \nu')(x_i - \nu')^\top \right \| \, ,
\]
where $\|\cdot\|$ denotes the spectral norm, or maximum eigenvalue.
\end{definition}

Note that that spectral sample reweighing is a \emph{worst-case} computational problem.
The basic optimization task underlying spectral sample reweighing is to find weights $w \in \mathcal{W}_{n,\epsilon}$ minimizing the spectral norm of the weighted second moment of $\{x_i - \nu\}_{i \in [n]}$---an $\alpha$-approximation is allowed to output instead $w$ in the slighly larger set $\mathcal{W}_{n, 3\epsilon}$ and may only minimize the spectral norm up to a multiplicative factor $\alpha$.
The parameter $\epsilon$ should be interpreted as the degree to which $w \in \mathcal{W}_{n, \epsilon}$ may deviate from the uniform distribution.

Our first result shows that robust and heavy-tailed mean estimation both reduce to spectral sample reweighing.

\begin{theorem}[Informal, see  \autoref{thm:robust}, \autoref{thm:reweight-heavy}]
  Robust and heavy-tailed mean estimation can both be solved with information-theoretically optimal error rates (up to constant factors) by algorithms which make one call to an oracle providing a constant-factor approximation to spectral sample reweighing (with $\epsilon = \epsilon_0$ a small universal constant) and run in additional time $\widetilde O(nd)$.
\end{theorem}

For robust mean estimation this reduction is implicit in \cite{diakonikolas2019robust} and others (see e.g. \cite{dong2019quantum}).
For heavy-tailed mean estimation the reduction was not previously known: we analyze it by a simple convex duality argument (borrowing techniques from \cite{cheng2019high,lecue2019robust}).
Our argument gives a new equivalence between two notions of a \emph{center} for a set of high-dimensional vectors---the spectral center considered in spectral sample reweighing and a more combinatorial notion developed by Lugosi and Mendelson in the context of heavy-tailed mean estimation \cite{lugosi2019sub}.
We believe this equivalence is of interest in its own right---see~\autoref{prop:p1} and \autoref{prop:2}.

We now turn our attention to Question 2.
We offer a unified approach to rigorously analyzing several well-studied algorithms by observing that each in fact instantiates a common strategy for \emph{online convex optimization}, and hence can be analyzed by applying a standard regret bound.
This leads to the following three theorems.
We first demonstrate that the filter, one of the first algorithms proposed for efficient robust mean estimation~\cite{diakonikolas2019robust,li2018principled,diakonikolas2017being,diakonikolas2019recent}, can be analyzed in this framework.
Specifically, we show:%, giving new results for gradient descent, a simplified proof of  the filter, and recovering the main result of \cite{dong2019quantum} on the quantum entropy filter.
\begin{theorem}[\cite{diakonikolas2019robust}, Informal, see \autoref{thm:filter}]
\label{thm:filter-intro}
There is an algorithm, \textsc{filter}, based on multiplicative weights, which gives a constant-factor approximation to spectral sample reweighing for  sufficiently small $\epsilon$, in time $\widetilde{O}(nd^2)$\footnote{We use $\widetilde O, \widetilde \Omega$ notation to hide polylogarithmic factors. Also, we remark that  a variant of our main algorithm achieves the optimal breakdown point of $1/2$; see \autoref{sec:breakdown}. }.
\end{theorem}

Previous approaches to analyzing the filter required  by-hand construction of potential functions to track the progress of the algorithm.
Our novel strategy to prove~\autoref{thm:filter-intro} demystifies the analysis of the filter by applying an out-of-the-box regret bound: the result is a significantly simpler proof than in prior work.  It allows us to capture robust mean estimation in both bounded covariance and sub-gaussian setting.

Moving on, we also analyze gradient descent, giving the following new result, which we also prove by applying an out-of-the-box regret bound. Although it gives weaker running-time bound than we prove for \textsc{filter}, the advantage is that the algorithm is vanilla gradient descent. (By comparison, the multiplicative weights algorithm of~\autoref{thm:filter-intro} can be viewed as a more exotic mirror-descent method.)
\begin{theorem}[Informal, see \autoref{thm:gdbound}]
\label{thm:gradient-intro}
  There is a gradient-descent based algorithm for spectral sample reweighing which gives a constant-factor approximation to spectral sample reweighing in ${O}(nd^2/\epsilon^2)$ iterations and $\widetilde O(n^2 d^3/\epsilon^2)$ time.
\end{theorem}
Prior work analyzing gradient descent for robust mean estimation required sophisticated tools for studying non-convex iterative methods \cite{cheng2020high}. 
Our regret-bound strategy shows for the first time that gradient descent solves heavy-tailed mean estimation, and that it solves robust mean estimation in significantly fewer iterations than previously known (prior work shows a bound of $\widetilde{O}(n^2 d^4)$ iterations in the robust mean estimation setting, where our bound gives ${O}(n d^2)$ iterations \cite{cheng2020high}). 

Finally, we demonstrate that the nearly-linear time algorithm for robust mean estimation in~\cite{dong2019quantum} fits into this framework as well. 
Thus, this framework captures state-of-the-art algorithms for robust mean estimation.
\begin{theorem}[\cite{dong2019quantum}, Informal, see \autoref{thm:mmw}]
\label{thm:que-intro}
There is an algorithm based on matrix multiplicative weights which gives a constant-factor approximation to spectral sample reweighing for sufficiently small $\epsilon$, in time $\widetilde{O}(nd \log(1/\epsilon))$.
\end{theorem}

\subsection{Related work}
 For robust mean estimation,~\cite{diakonikolas2019robust,lai2016agnostic}  give the first polynomial-time algorithm with optimal (dimension-independent) error rates. Their results have been further improved and generalized by a number of works~\cite{balakrishnan2017computationally, diakonikolas2017being, diakonikolas2018robustly, diakonikolas2018list, hopkins2018mixture, steinhardt2018resilience, dong2019quantum, diakonikolas2019sever, cheng2019faster, diakonikolas2019efficient}. See~\cite{diakonikolas2019recent} for a complete survey.

The first (computationally inefficient) estimator to obtain optimal confidence intervals for heavy-tailed distributions in high dimensions is given by~\cite{lugosi2019sub}; this construction was first made algorithmic by~\cite{hopkins2018mean}, using the Sum-of-Squares method. Later works~\cite{cherapanamjeri2019fast,lecue2019robust, lei2019fast} significantly improve the run-time: the fastest known algorithm runs in time $\widetilde O(n^2d)$.

Analyses of the \textsc{Filter} algorithm are scattered around the literature~\cite{diakonikolas2019robust,li2018principled,diakonikolas2017being,diakonikolas2019recent}.  The variant of \textsc{Filter} we present here is based on a soft downweighting procedure first proposed by~\cite{steinhardt2018robust}. However, no prior work analyzes \textsc{Filter} through the lens of regret minimization or points out a connection with the heavy-tailed setting.

Prior works~\cite{lecue2019robust, prasad2019unified, lugosi2019robust} have proposed unified approaches to heavy-tailed and robust mean estimation. In particular, \cite{lecue2019robust} observes a robustness guarantee of~\cite{lugosi2019sub}, originally designed for the heavy-tailed setting. However, these works do not distill a meta-problem or obtain the analysis via duality. In addition, although it matches the fastest-known running time in theory, the algorithm of~\cite{lecue2019robust} is based on semidefinite programming, rendering it relatively impractical.
Some constructions from~\cite{prasad2019unified, lugosi2019robust} are not known to be computationally tractable.

In a concurrent and independent work,~\cite{zhu2020robust} also studies the spectral sample reweighing problem (in the context of robust mean estimation), and provides an analysis of filter-type algorithms based on a regret bounds.
The argument of~\cite{zhu2020robust} relies on a technical optimization landscape analysis, which our arguments avoid. 
The framework of~\cite{zhu2020robust} can be extended to robust linear regression and covariance estimation; it is unclear our techniques extend similarly.
Their work  also proves  an optimal breakdown point analysis of filter-type algorithm for robust mean estimation. We obtain the same result (see \autoref{sec:breakdown}) with an arguably less sophisticated proof.
Lastly,~\cite{zhu2020robust} does not discuss the heavy-tailed setting.

Another concurrent and independent work,~\cite{diakonikolas2020outlier}, also shows that the filter and non-convex gradient descent obtain optimal rates in the robust and heavy-tailed settings.
The authors also identify a general stability-based condition under which robust mean estimation algorithms achieve optimal rates in the heavy-tailed setting.

\subsection{Organization}
We formally introduce the spectral sample reweighing problem and analyze an algorithm based on the \textsc{Filter} algorithm in Section \ref{sec:mwu}. We show how this primitive can be immediately used to solve the robust mean estimation problem in Section~\ref{sec:robust}. Then in Section~\ref{sec:centrality} we introduce the duality theorem that connect two notions of centrality. The result is used further in Section~\ref{sec:heavy-tail}, where we show how to leverage  the duality for heavy-tailed mean estimation.

\section{Preliminaries} 
For a set of $n$ real values $\alpha_i$, we let $\textsf{median} \left( \{ \alpha_i\}_{i=1}^n \right)$ to denote its median.
For a matrix $A$, we use $\|A\|,\|A\|_2$ to denote the spectral norm of $A$ and $\Tr(A)$ its trace. For a vector $v$, $\|v\|_p$ denotes the $\ell_p$ norm. 
We denote the all-one vector of dimension $k$ by $\I_k$. 
For vectors $u,v$, we denote the entrywise product by $u\odot v$; that is, the vector such that $w_i = u_i \cdot v_i$ for each $i$.
For PSD matrices $A,B$, we write $A \preceq B$ if $B -A$ is PSD. Density matrices refer to the set of PSD matrices with unit trace. For any symmetric
matrix $A \in \Real^{d\times d}$, let exp$(A)$ denote the   matrix exponential of $A$. For a weight vector $w$ such that $0 \leq w_i \leq 1$ and point set $\{x_i\}_{i=1}^n$, we define $\mu(w) = \sum_{i=1}^n w_i x_i$ and $M(w) = \Sigma_w = \sum_{i=1}^n w_i (x_i - \mu(w) ) (x- \mu(w))^\top$. 

\begin{definition}[approximate top eigenvector]
For any PSD matrix $M$ and $c \in (0,1)$, we say that a unit vector $v$ is a $c$-approximate largest eigenvector of $M$ if $v^T Mv \geq c \|M\|_2$. 
\end{definition}
For a PSD matrix $M$, we let \textsc{ApproxTopEigenvector}$(M, c, \alpha)$ to denote an approximation scheme that outputs a (unit-norm) $c$-approximate largest eigenvector of $M$ with a failure probability of at most $\alpha$. The classic power method achieves such guarantee with a run-time of $O\left(\tfrac{1}{1-c} nd \log(1/\alpha)\right)$, when $M$ is given in a factored form $M= X^\top X$, for $X\in \Real^{n\times d}$.

\begin{definition}[Kullback–Leibler divergence]
For probability distributions $p,q$ over $[n]$, the KL divergence from $q$ to $p$ is defined as $\text{KL}(p || q) =\sum_{i=1}^n p(i) \log \tfrac{p(i)}{q(i)}$.
\end{definition}
\begin{definition}[total variation distance]
For probability distributions $p,q$, the total variation distance is defined as $\text{TV}(p,q) = \sup_E |p(E) - q(E)| =\frac12 \|p-q\|_1$, where the supremum is over the set of measurable events.
\end{definition}
 
We use $\Delta_n$ to denote the set of probability distributions over $[n]$ and write $\mathcal U_n$ for  the uniform distribution over $[n]$. We use $i\sim I$ to denote $i$ drawn uniformly from an index set $I \subseteq [n]$. Throughout, we define  $\mathcal W_{n,\epsilon} =\{ w\in \Delta_n : w_i \leq \frac{1}{(1-\epsilon)n}\}$ to be the discrete distributions over $[n]$ with bounded $\ell_\infty$ norm; we call the set \textit{good weights}.
\section{The Meta-Problem and a Meta-Algorithm}
\label{sec:mwu}
We now define the meta-problem, which we call \textit{spectral sample reweighing}, that underlies both the adversarial and heavy-tailed models.
We put it as a promise problem.
\begin{definition}[$(\alpha,\epsilon)$-spectral sample reweighing] \label{def:ssr}
Let $\epsilon \in (0,1/2)$. The spectral sample reweighing problem is specified by the following.
\begin{itemize}
     \item \textit{Input}: $n$ points $\{x_i\}_{i=1}^n$ in $\mathbb R^d$ and $\lambda \in \mathbb R$.
     \item \textit{Promise}: There exists a $\nu\in \mathbb R^d$ and  a set of good weights $w \in \mathcal{W}_{n
    ,\epsilon}$ such that 
\begin{equation}\label{eq:p} \tag{$\dagger$}
    \sum_{i=1}^n w_{i}\left(x_{i}- \nu\right)\left(x_{i}- \nu\right)^{\top}\preceq \lambda I.
\end{equation} 
\item \textit{Output}:   A set of good weights $w'\in \mathcal W_{n,3 \epsilon}$ and $\nu' \in \mathbb R^d$ that satisfies the   condition above, up to the factor of $\alpha \geq 1$:
\begin{equation}\label{eq:goal}
    \sum_{i=1}^n w'_{i}\left(x_{i}- \nu'\right)\left(x_{i}- \nu'\right)^{\top}  \preceq  \alpha \lambda I.
\end{equation}
 \end{itemize}
\end{definition}
To provide some   intuition, the goal here is to find a set of weights $\{w_i\}_{i=1}^n$, close to the uniform distribution on $[n]$, and a center $\nu$ such that by weighting by $w$ and centering by $\nu$, the covariance  is bounded, under the promise that such a set of weights exists. We will refer to our promise~\eqref{eq:p} as a \textit{spectral centrality} assumption.  
\paragraph{Solving spectral sample reweighing} The main result of this section is an efficient algorithm that achieves a constant factor approximation for the spectral sample reweighing problem. 
\begin{theorem}[\cite{diakonikolas2019robust} spectral sample reweighing via filter]\label{thm:filter}
Let  $\{x_i\}_{i=1}^n$ be $n$ points in $\mathbb R^d$ and $\epsilon \in (0,1/10]$.
Suppose there exists $\nu\in \mathbb R^d$ and  $w \in \mathcal{W}_{n
    ,\epsilon}$ such that 
    \[
    \sum_{i=1}^n w_{i}\left(x_{i}- \nu\right)\left(x_{i}- \nu\right)^{\top}\preceq \lambda I
    \]
for some   $\lambda  > 0$.   
Then, given $\{x_i\}_{i=1}^n , \lambda, \epsilon$ and a failure rate $\delta$, there is an algorithm that finds   $w'\in \mathcal W_{n,3\epsilon}$ and   $\nu' \in \mathbb R^d$ such that $$
    \sum_{i=1}^n w'_{i}\left(x_{i}- \nu'\right)\left(x_{i}- \nu'\right)^{\top}  \preceq 60\lambda I,$$ with probability at least $1-\delta$. 
    
    The algorithm runs in $O(d)$ iterations and $\widetilde O\left(nd^2\log (1/\delta)\right)$ time in total.
\end{theorem}
Our algorithm is  a a multiplicative weights-style
procedure.  
In particular, the output center $\nu'$ will be a weighted average of  the points $\{x_i\}_{i=1}^n$. 
The algorithm starts with the uniform weighting and iteratively downweights points which are causing the empirical covariance to have a large  eigenvalue. To ensure that we always maintain a set of good weights, we project the weights onto the   set $\mathcal W_{n,\epsilon}$ at the end of each iteration, according to KL divergence.  
For technical reasons, the algorithm  also requires a \textit{width} parameter $\rho$. It suffices to set it as the squared diameter  of the input points $\{x_i\}_{i=1}^n$, and it can be bounded by $O(d\lambda/\epsilon)$ by a simple pruning argument~(\autoref{lem:diameter} and \autoref{lem:naive-prune}).

The algorithm should be seen as a variant of the \textsc{Filter} algorithm, due to Diakonikolas, Kamath, Kane,  Li,  Moitra, and Stewart~\cite{diakonikolas2019robust}. The procedure we present here most resembles a  more streamlined version later by Steinhart~\cite{steinhardt2018robust}.  However, neither   formulated the problem quite this way or gave this analysis. 
Instead, we will re-analyze the algorithm  for the purpose of spectral sample reweighing and in a different manner than previously done in the literature.

\begin{algorithm}[ht!] \label{alg:filter}
    \caption{Multiplicative weights   for spectral sample reweighing~(\autoref{def:ssr})}
    \KwIn{A set of   points $\{x_i\}_{i=1}^n$, an iteration count $T$, and parameter $\rho,\delta$}
\KwOut{A point $\nu\in \mathbb R^d$  and weights $w\in \mathcal{W}_{n,\epsilon}$.}
\vspace{10pt}
Let  $w^{(1)} = \frac{1}{n} \I_n$ and $\eta = 1/2$. \\
\For{$t$ from $1$ to $T$}{
    Let $\nu^{(t)} = \sum_i w^{(t)}_i x_i$, $M^{(t)} = \sum_i w^{(t)}_i
(x_i- \nu^{(t)})(x_i - \nu^{(t)})^T$.\\
Compute $v^{(t)} = \textsc{ApproxTopEigenvector}(M^{(t)}, 7/8, \delta/T)$.\\
Compute $\tau^{(t)}_i =\left \langle v^{(t)}, x_i- \nu^{(t)} \right\rangle^2$ for each $i$.\\
    Set $w_i^{(t+1)} \leftarrow w_i^{(t)}\left(1-\eta  \tau^{(t)}_i/\rho\right)$ for each $i$.\\
    Project $w^{(t+1)}$ onto the set of good weights $\mathcal W_{n,\epsilon}$ (under KL divergence):
    \[
    w^{(t+1)} \leftarrow \argmin_{w\in \mathcal W_{n,\epsilon}}\,\, \text{KL}\left(w|| w^{(t)}\right).
    \]
    %\Snote{Might need some clarification on what this projection means.}
}
\Return $\nu^{(t^*)}, w^{(t^*)}$, where $t^* = \argmin_t \|M^{(t)}\|$.
\end{algorithm}

\begin{lemma}[analysis of filter]\label{lem:filter}
Let  $\epsilon \in (0, 1/10]$ and $\{x_i\}_{i=1}^n$ be $n$ points in $\mathbb R^d$.
Suppose there exists $\nu\in \mathbb R^d$ and  $w \in \mathcal{W}_{n
    ,\epsilon}$ such that $$
    \sum_{i=1}^n w_{i}\left(x_{i}- \nu\right)\left(x_{i}- \nu\right)^{\top}\preceq \lambda I$$
for some   $\lambda  > 0$.   
Then, given $\{x_i\}_{i=1}^n$, a failure rate $\delta$ and $\rho$ such that $\rho \geq \tau_i^{(t)}$ for all $i$ and $t$, Algorithm \ref{alg:filter} finds   $w'\in \mathcal W_{n,\epsilon}$ and   $\nu' \in \mathbb R^d$ such that 
\begin{equation}\label{eqn:lem-goal}
    \sum_{i=1}^n w'_{i}\left(x_{i}- \nu'\right)\left(x_{i}- \nu'\right)^{\top}  \preceq 60\lambda I,
\end{equation} with probability at least $1-\delta$. 
    
    The algorithm terminates in   $T = O(\rho \epsilon / \lambda)$ iterations. Further, if $T =O(\text{poly}(n,d))$, then each iteration takes $\widetilde O(nd\log (1/\delta))$ time.
\end{lemma}
We first see how to prove~\autoref{thm:filter} via~\autoref{lem:filter}. Note that it requires   to bound the width parameter $\rho$. 
To ensure the condition $\rho \geq \tau^{(t)}_i$ for all $i$ and $t$, observe    that as   $\|v^{(t)}\|=1$, we have $$\tau_i^{(t)} =  \left \langle v^{(t)}, x_i- \nu^{(t)} \right\rangle^2   \leq \| x_i - \nu^{(t)}\|^2.$$
Also,  since $\nu^{(t)}$ is a convex combination of $\{x_i\}_{i=1}^n$, we can set $\rho$ to be the squared diameter of  the input data $\{x_i\}_{i=1}^n$.  
As the first step, we  show that a $(1-2\epsilon)$ fraction of the points lie within a ball of radius $\sqrt{d\lambda/\epsilon}$ under the spectral centrality condition. Then a (folklore) pruning procedure can be used to extract such set. 

\begin{lemma}[diameter bound]\label{lem:diameter}
Let  $\{x_i\}_{i=1}^n$ be $n$ points in $\mathbb R^d$.
Suppose there exists $\nu\in \mathbb R^d$ and  $w \in \mathcal{W}_{n
    ,\epsilon}$ such that 
    $
    \sum_{i=1}^n w_{i}\left(x_{i}- \nu\right)\left(x_{i}- \nu\right)^{\top}\preceq \lambda I
    $
for some   $\lambda  > 0$ and  $\epsilon \in (0,1/2)$. Then there  exists a ball of radius $\sqrt{d\lambda/\epsilon}$ that contains at least $r=(1-2\epsilon)n$ points of $\{x_i\}_{i=1}^n$.
\end{lemma}
The proof of the lemma can be found in Appendix~\ref{sec:diameter}

\begin{lemma}[folklore; see~\cite{dong2019quantum}]
\label{lem:naive-prune}
Let $\epsilon < 1/2$ and $\delta > 0$.
Let $S \subset \Real^d$ be a set of $n$ points. Assume there exists a ball $B$ of radius $r$ and a subset $S' \subseteq S$ such that $|S'| \geq (1 - \epsilon)n$  and $S' \subset B$.
Then there is an algorithm $\textsc{Prune}(S, r, \delta)$ that runs in time $O(n d \log 1 / \delta)$ and with probability $1 - \delta$ outputs a set  $R \subseteq S$ so that $S' \subseteq R$, and $R$ is contained in a ball of radius $4r$.
\end{lemma}
Using the lemmas above, we can immediately prove the main theorem.
\begin{proof}[Proof of~\autoref{thm:filter}]
Given $S= \{x_i\}_{i=1}^n, \lambda$ and $\epsilon$, we first run the \textsc{Prune}($S, r, \delta/2$) algorithm, with $r= \sqrt{d\lambda /\epsilon}$.  By~\autoref{lem:diameter}, the spectral centrality condition~\eqref{eq:p} implies  there  exists a ball of radius $r$   containing at least $(1-2\epsilon)n$ points of $S$. 
Therefore, \autoref{lem:naive-prune} guarantees that it will return a set $R \subseteq S$ of  at least $(1-2\epsilon)n$ points contained in a ball of radius $4r$.  
Hence by \autoref{lem:filter}, given $R$,  $\rho = 16d\lambda/\epsilon$ and failure rate $\delta / 2$, \autoref{alg:filter} finds  $w'\in \mathcal W_{|R|,\epsilon}$ and   $\nu' \in \mathbb R^d$ such that $$
    \sum_{i\in R} w'_{i}\left(x_{i}- \nu'\right)\left(x_{i}- \nu'\right)^{\top}  \preceq 60\lambda I,$$ with probability at least $1-\delta/2$.  Let $w_i''= w_i'$ if $i\in R$ and $w_i'' =0$ otherwise. since $\frac{1}{(1-\epsilon)(1-2\epsilon) }\leq  \frac{1}{1-3\epsilon}$ for $\epsilon < 1/3$, 
     we have $w'' \in \mathcal{W}_{n,3\epsilon}$
    Moreover, $\sum_{i=1}^n w''_{i}\left(x_{i}- \nu'\right)\left(x_{i}- \nu'\right)^{\top}  \preceq 60\lambda I$, as desired.
    
    The overall procedure succeeds with probability at least $1-\delta$ by a union bound, since \autoref{alg:filter} and \textsc{Prune} are both set up to have a failure rate at most $\delta/2$. Now for the run-time, \textsc{Prune}($S,r,\delta$) takes $O(nd\log(1/\delta))$ by~\autoref{lem:naive-prune}. Moreover, by \autoref{lem:filter}, \autoref{alg:filter} runs in time $\widetilde O(nd \log(1/\delta) \cdot T)$ time, with $T= O(\rho \epsilon  / \lambda)$ being the iteration count. Since $\rho = 16d\lambda/ \epsilon$, we have $T= O(d)$, and this immediately yields the desired runtime.
\end{proof}

\paragraph{Analysis via regret minimization} Now it remains to analyze~\autoref{alg:filter}, proving~\autoref{lem:filter}. We will cast the algorithm under the framework of regret minimization using multiplicative weights update (MWU). To see that, we consider $\{x_i\}_{i=1}^n$ as the set of actions, $w^{(t)}$ as our probability distribution over the actions at time $t$, and we receive a loss vector $\tau^{(t)}$ each round. The weights are updated in a standard fashion. Then, to ensure that the weights lie in the constraint set $\mathcal{W}_{n,\epsilon}$, we perform a projection step. (Note that the algorithm is implementing both the player and the adversary.)
The following is a classic regret bound of MWU for the online linear optimization problem.

\begin{lemma}[regret bound \cite{arora2012multiplicative}]\label{lem:regret}
Suppose $\rho \geq \tau^{(t)}_i$ for every $t$ and $i$. Then for \textit{any}  weight $w\in \mathcal{W}_{n,\epsilon}$, \autoref{alg:filter} satisfies that 
    \begin{align}
        \frac{1}{T}\sum_{t=1}^T \left\langle w^{(t)}, \tau^{(t)} \right\rangle \leq\frac{1}{T}(1+\eta)
        \sum_{t=1}^T \left\langle w,
        \tau^{(t)}\right \rangle  + \frac{\rho \cdot \text{KL}(w|| w^{(1)})}{T \eta },
    \end{align} 
    for any choice of step size $\eta \leq  1/2$.
    \end{lemma}
\noindent
     In addition, we claim the following  lemma and delay its proof to the appendix~(\autoref{lem:ssw}).  
     \begin{lemma}\label{lem:ss}
          Under the centrality promise~\eqref{eq:p}, for any $w'\in \mathcal{W}_{n,\epsilon}$,  
    \begin{equation}
        \|\nu- \nu(w')\| \le  \frac{1}{1-\sqrt{2\epsilon}} \left(   \sqrt {2\lambda} +   \sqrt{2\epsilon\|M(w')\|}\right),
    \end{equation}
    where $\nu(w') = \sum_{i}w'_i x_i$ and $M(w') = \sum_i w'_i (x_i - \nu(w'))(x_i - \nu(w'))^\top$.
     \end{lemma}
    This type of inequality is generally known as the \textit{spectral signature} lemma from the recent algorithmic robust statistics literature; see \cite{li2018principled, diakonikolas2019recent}. 

With these technical ingredients, we are now ready to analyze the algorithm.
\begin{proof}[Proof of~\autoref{lem:filter}]
Notice first that since $v^{(t)}$ is a $7/8$-approximate largest eigenvector of $M^{(t)}=\sum_i w^{(t)}_i
(x_i- \nu^{(t)})(x_i - \nu^{(t)})^T$, then for all $t$,
\begin{equation} \label{eq:lhs}
    \sum_i w_i^{(t)} \tau_i^{(t)} = \sum_i w_i  \left \langle v^{(t)}, x_i- \nu^{(t)}
    \right\rangle^2 = v^{(t)\top } M^{(t)}v^{(t)} \geq \frac 7 8 \left\|M^{(t)}\right\|_2.
\end{equation}
Let $w$ be the good weights that satisfies our centrality promise \eqref{eq:p}. Summing over the $T$ rounds and applying the  the regret bound~(\autoref{lem:regret}), we obtain that
\begin{equation*}
     \frac{7}{8T} \sum_{t=1}^T \left\|M^{(t)}\right\|_2
  \leq   \frac{1}{T} \sum_{t=1}^T \left \langle w^{(t)}  , \tau^{(t)}\right\rangle \leq (1+\eta) \frac{1}{T}     \sum_{t=1}^T\left \langle w , \tau^{(t)}\right \rangle +\frac{\rho\cdot  \text{KL}(w|| w^{(1)}) }{T\eta}.
\end{equation*}
 The KL term can be  bounded because $w$ and $w^{(1)}$ are both close to uniform. 
Indeed, it is a simple calculation to verify that $ \text{KL}(w|| w^{(1)}) \leq 5\epsilon$, using the fact $w_i \leq 1/(1-\epsilon)n$ (\autoref{lem:kl}).
Plugging in $\eta = 1/2$, we get 
\begin{equation}\label{eq:imm}
     \frac{7}{8T} \sum_{t=1}^T \left\|M^{(t)}\right\|_2
  \leq \frac{3}{2T}     \sum_{t=1}^T\left \langle w , \tau^{(t)}\right \rangle +\frac{10\epsilon\rho}{T}.
  \end{equation}
Our eventual goal is to bound this by $O(\lambda)$. Note that the second term is easy to control---just set  $T = \Omega(\rho\epsilon/\lambda)$, and this will determine the iteration count and thus the runtime.

The remaining is mostly tedious calculations to bound the first term. The reader can simply skip forward to~\eqref{eq:concl}. For those interested: we proceed by expanding the first term on the right-hand side,
\begin{align}
    \frac{3}{2T}
    \sum_{t=1}^T\left \langle w , \tau^{(t)}\right\rangle&=  \frac{3}{2T} \sum_{t=1}^T  \sum_{i=1}^n w_i \left\langle x_i - \nu^{(t)}
        ,v^{(t)}\right \rangle^2\label{eq:aa1}\\
        &= \frac{3}{2T} \sum_{t=1}^T \sum_{i=1}^n   w_i \left(\left\langle x_i - \nu
    ,v^{(t)}\right \rangle^2 +  \left\langle\nu- \nu^{(t)}, v^{(t)}\right \rangle^2\right)\label{eq:aa2}\\
      &\leq  \frac32 \lambda +  \frac{3}{2T} \sum_{t=1}^T \left\langle\nu- \nu^{(t)}, v^{(t)}\right \rangle^2\label{eq:aa3}\\
    &\leq \frac32 \lambda + \frac{3}{2T} \sum_{t=1}^T  \left\| \nu - \nu^{(t)}\right\|_2^2,\label{eq:aa4}
\end{align}
where~\eqref{eq:aa1} is by the definition that $\tau^{(t)}_i =\left \langle v^{(t)}, x_i- \nu^{(t)} \right\rangle^2$, \eqref{eq:aa2} uses the definition of $\nu^{(t)}$, \eqref{eq:aa3} follows from the spectral centrality assumption~\eqref{eq:p}, and \eqref{eq:aa4} is by the fact that $\|v^{(t)}\|=1$.  Since $\nu^{(t)} = \sum_{i=1}^n w^{(t)}_i x_i$, we can apply~\autoref{lem:ss} to bound $\| \nu - \nu^{(t)}\| $ and it follows that
\begin{align*}
    \frac{3}{2T} \sum_{t=1}^T  \left\| \nu - \nu^{(t)}\right\|_2^2 \leq \frac{3}{2T}\left(\sum_{t=1}^T \frac{25}{2}\lambda + \frac{1}{3}\left\|M^{(t)}\right\|_2\right),
\end{align*}
for  $\epsilon \leq 1/10$.
Plugging the bound into~\eqref{eq:aa4}, we obtain 
\begin{align}\label{eq:right-2}
    \frac{3}{2T}
    \sum_{t=1}^T \langle w , \tau^{(t)}\rangle&\leq \frac32 \lambda + \frac{3}{2T}\left(\sum_{t=1}^T \frac{25}{2}\lambda + \frac{1}{3}\left\|M^{(t)}\right\|_2\right) = \frac {81}{4}\lambda+ \frac{1}{2T}\sum_{t=1}^T \left\|M^{(t)}\right\|_2.
\end{align}
Finally, substituting this back into~\eqref{eq:imm}, we see that 
\begin{align}\label{eq:concl}
     \frac{7}{8T} \sum_{t=1}^T \left\|M^{(t)}\right\|_2
  \leq  \frac {81}{4}\lambda+ \frac{1}{2T}\sum_{t=1}^T \left\|M^{(t)}\right\|_2 + \frac{10\epsilon\rho}{T}.
\end{align}
Now if we set  $T = 10\rho \epsilon / \lambda$, then the last term is $\lambda$. Rearranging  yields that  $  \frac{1}{T} \sum_{t=1}^T \left\|M^{(t)}\right\|_2 \leq 60\lambda$. This shows that within $T =O( \rho\epsilon/ \lambda)$ iterations we have achieved our goal~\eqref{eqn:lem-goal}.  
%\fred{Is there a simple way to see that the $M^{(t)}$ are monotonically decreasing, so that output the last iterate is the right thing to do?}

Now it remain to argue the cost of each iteration.
For approximating the largest eigenvector, the well-known power method computes a constant-approximation  in $ O(nd\log (1/\alpha))$ time with a failure probability  at most $\alpha$~\cite{kuczynski1992estimating}. We   set $\alpha  = \delta / T$, and an application of union bound implies that all the $T$ calls to the power method jointly succeed with probability at least $1-\delta$. This gives a total run-time of $\widetilde O(nd\log (1/\delta))$, since $T =O(\text{poly}(n,d))$, and bounds the overall failure probability of the algorithm by $\delta$.
Finally, we remark that  the  KL projection onto $\mathcal{W}_{n,\epsilon}$   can be computed exactly in $O(n)$  time, by the deterministic procedures in \cite{herbster2001tracking, warmuth2008randomized}. This completes the run-time analysis.
\end{proof}

\paragraph{Faster algorithm}
Under the same assumptions, the spectral sample reweighing problem can be solved in $\widetilde O(nd \log (1/\delta))$ time, by adapting a \textit{matrix} multiplicative weight scheme, due to Dong, Hopkins and Li~\cite{dong2019quantum}. 
The algorithm and its analysis generally follow from the proofs therein. The details can be found in Appendix~\ref{sec:mmwu}. 

As we will see soon, applying this procedure  directly  match the fastest known algorithms for both robust and heavy-tailed settings. 

\paragraph{Gradient descent analysis} 
As we argued, \autoref{alg:filter} is essentially an online linear optimization scheme, with the objective of minimizing $\sum_{t=1}^T \langle w^{(t)}, \tau^{t}\rangle$. It is known that the multiplicative weights   rule employed here can be seen an entropic mirror descent update~\cite{pmlr-v32-steinhardtb14}. Therefore, it is natural to ask whether an additive update/gradient descent procedure would solve the problem as well. In Appendix~\ref{sec:gd}, we provide such an analysis (\autoref{thm:gdbound}). More importantly, the resulting scheme is equivalent of the gradient descent algorithm analyzed by~\cite{cheng2020high}. Our analysis improves upon the iteration complexity  from their work (in the concrete settings of robust mean estimation, under bounded second moment and sub-gaussian distributions).

\section{Estimation under Corruptions}\label{sec:robust}
We now apply \autoref{alg:filter} for the robust mean estimation problem.
We focus on the bounded second moment distributions, where   \autoref{alg:filter} can be invoked in a   black-box fashion. A slight variant of it can be used for the sub-gaussian setting, where we achieve a more refined analysis; see Appendix \ref{sec:subg-filter}.

The problem is formally defined below.
\begin{definition}[robust mean estimation]
Given a distribution $D$ over $\mathbb R^d$ with bounded covariance and a parameter $ 0\leq \epsilon <1/2$, the adversary draws $n$ i.i.d.\ samples $D$, inspects the samples, then   removes at most $\epsilon n$ points and replaces them with arbitrary points. We call the resulting dataset $\epsilon$-corrupted (by an adaptive adversary).

The goal is to estimate the mean of $D$ only given the $\epsilon$-corrupted set of samples. 
\end{definition} 
Using a meta-algorithm for approximating the spectral sample reweighing problem, we will show the following. In particular, using Algorithm \ref{alg:filter}  matches the run-time and statistical guarantee of the original \textsc{Filter} algorithm.
\begin{theorem}[robust mean estimation via sample reweighing]\label{thm:robust}
Let $D$ be a distribution  over $\mathbb R^d$ with mean $\mu$ and covariance $\Sigma \preceq  \sigma^2 I$ and $\epsilon\leq 1/10$.  Given an $\epsilon$-corrupted set of    $n = \Omega(d\log d / \epsilon)$ samples, there is an algorithm that runs in time $\widetilde O(nd^2)$ that with constant  probability outputs an estimate $\widehat \mu$ such that $\|\widehat \mu - \mu\| \leq O(\sigma \sqrt{\epsilon})$.

Further, the algorithm is via a black-box application of Algorithm~\ref{alg:filter}, which can be replaced by any constant approximation algorithm for the spectral sample reweighing problem (\autoref{def:ssr}).
\end{theorem}
Information-theoretically, \autoref{thm:robust} is  near optimal. 
     It is known that the sample complexity of $d\log d/ \epsilon$ is tight, only up to the log factor. The estimation error $O(\sqrt{\epsilon})$  is tight up to constant factor.
 
Our analysis requires a set of   deterministic conditions to hold for the input, which follow from Lemma A.18 of~\cite{diakonikolas2017being}.  This is meant to obtain the desired spectral centrality condition and to bound the final estimation error.
\begin{lemma}[deterministic conditions~\cite{diakonikolas2017being}]\label{lem:second}
Let $S$ be an $\epsilon$-corrupted set of $\Omega(d\log d /\epsilon)$ samples from $D$ with mean $\mu$ and covariance $\Sigma \preceq I$. With high constant probability, $S$ contains a subset $G$ of size  at least $(1-\epsilon) n$ such that
 \begin{align}
     &\| \mu - \mu_G\| \le  O(\sqrt{\epsilon})\label{eq:meang}\\
     & \left\|\frac{1}{\left|G\right|} \sum_{i \in G}\left(x_{i}-{\mu_G}\right)\left(x_{i}-{\mu_G}\right)^{\top}\right\|_{2} \le O(1),\label{eq:covg}
 \end{align}
 where ${\mu_G}=\frac{1}{\left|G\right|} \sum_{i \in G} x_{i}$.
 %Moreover, the diameter of the remaining samples can be bounded by $O\left(\sqrt{d/\epsilon}\right)$.
\end{lemma}
We now prove the main result of this section---using the meta-algorithm to solve the robust mean estimation problem. 
Observe that it suffices to prove the theorem with $\sigma^2 = 1$. Without loss of generality, we can
 first  divide every input  sample by $\sigma$,  execute the algorithm and then multiply the output by $\sigma$.
\begin{proof}[Proof of \autoref{thm:robust}]
First, we check that the centrality promise~\eqref{eq:p} is satisfied. This would ensure that we are in the setting of the spectral sample reweighing problem so that the meta-algorithm applies. 
Assume the conditions from~\autoref{lem:second}.
Then suppose we let $w_i = 1/|G|$ if $x_i \in G$ and $w_i = 0$ otherwise, so we have that $w\in \mathcal{W}_{n,\epsilon}$, and let $\nu=\mu_G$ and $\lambda = O(1)$. Observe that~\eqref{eq:covg} is exactly the spectral centrality condition~\eqref{eq:p} . Then we can apply~\autoref{thm:filter} and obtain that the algorithm  will find $\nu' \in \mathbb R^d$ and $w'\in \mathcal W_{n,3\epsilon}$ such that 
\begin{align*}
    M(w')  :=  \sum_{i=1}^n w'_{i}\left(x_{i}- \nu'\right)\left(x_{i}- \nu'\right)^{\top}  \preceq O(1)\cdot I
\end{align*}
Furthermore, by definition of the algorithm, $\nu'$ is a weighted average of the points $\{x_i\}_{i=1}^n$; that is, $\nu' = \nu(w') = \sum_{i=1}^n w'_i x_i$.  This allows us again to apply the  spectral signature lemma. In particular, \autoref{lem:robustspec} implies
\begin{align*}
   \|\mu_G - \nu' \| \leq \frac{1}{1-6\epsilon} \left( \sqrt {6\epsilon\lambda} +   \sqrt{3\epsilon\|M(w')\|}\right) = O\left(\sqrt{\epsilon}\right)
\end{align*}
  since $\lambda= O(1)$ and $\|M(w')\|= O(1)$. 
Finally, by triangle inequality and~\eqref{eq:meang},
\begin{align*}
     \|\mu- \nu' \| \leq  \|\mu_G - \nu' \| + \| \mu - \mu_G\| \leq  O(\sqrt{\epsilon}).
\end{align*}
Therefore, the output $\nu'$ estimates the true mean up to an error of $O(\sqrt{\epsilon})$, as desired. 

Finally, the run-time guarantee follows directly from the statement of~\autoref{thm:filter}, since we apply the meta-algorithm   in a black-box fashion. This completes the proof.
\end{proof}

\paragraph{Optimal breakdown point}
In \autoref{sec:breakdown}, we show that a variant of the filter algorithm can be used to achieve the optimal breakdown point of $1/2$. The result also appeared in a concurrent work~\cite{zhu2020robust},  with an argubly  more sophisticated proof.

\paragraph{Other algorithms}
To improve the computational efficiency, applying the same argument and using the matrix multiplicative weight algorithm (\autoref{thm:mmw}), we can obtain a near linear time algorithm, which matches the fastest known algorithm for robust mean estimation~\cite{dong2019quantum,cheng2019faster}. 
\begin{corollary}[faster robust mean estimation~\cite{dong2019quantum}]\label{thm:robust2}
Let $D$ be a distribution  over $\mathbb R^d$ with mean $\mu$ and covariance $\Sigma \preceq  \sigma^2 I$ and $\epsilon$ be a sufficiently small constant.  Given an $\epsilon$-corrupted set of    $n = \Omega(d\log d / \epsilon)$ samples, there is a matrix multiplicative update algorithm that runs in time $\widetilde O(nd)$ and with constant  probability computes an estimate  of error $O(\sigma \sqrt{\epsilon})$.
\end{corollary}

Since $\lambda = O(1)$ in the robust mean estimation problem under bounded covariance (\autoref{lem:second}),  our analysis of the gradient descent algorithm (\autoref{thm:gdbound}) implies the following.
\begin{corollary}[robust mean estimation via gradient descent]\label{thm:robust-gd}
Let $D$ be a distribution  over $\mathbb R^d$ with mean $\mu$ and covariance $\Sigma \preceq  \sigma^2 I$ and $\epsilon$ be a sufficiently small constant.  Given an $\epsilon$-corrupted set of    $n = \Omega(d\log d / \epsilon)$ samples, there is a  gradient-descent based algorithm that computes an estimate of error $O(\sigma\sqrt{\epsilon})$ with constant probability in $\widetilde O(nd^2/\epsilon^2)$ iterations.\footnote{The $1/\epsilon$ dependence 
in the run-time can be   removed by a simple bucketing trick due to~\cite{lecue2019robust}; also see Lemma B.1 of~\cite{dong2019quantum}.} 
\end{corollary}
A variant of the gradient descent-based algorithm can be used for robust mean estimation in the sub-gaussian setting as well; see Appendix~\ref{sec:subg-gd}.
\section{Equivalent Notions of Centrality}\label{sec:centrality} 
In this section, we   prove a  duality statement that connects the setting of heavy-tailed and robust estimation.  In particular, we will show that the following two (deterministic) notions of a \textit{center} $\nu$ for points $\{x_i\}_{i=1}^k$ are essentially  equivalent. We call them \textit{spectral} and \textit{combinatorial} center.
The former is the requirement  that showed up  first in the original formulation of the spectral sample reweighing problem (\autoref{def:ssr}) and then in dealing with adversarial corruptions. The latter  will  yield the right notion of high-dimensional median for estimating the mean of heavy tailed data, now known as the \textit{Lugosi-Mendelson estimator}, due to~\cite{lugosi2019sub}. 

In the following, let  $\{x_i\}_{i=1}^k$ be a set of $k$ points in $\mathbb R^d$.
    \paragraph{Spectral center}  Recall that our meta-problem of spectral sample reweighing (\autoref{def:ssr}) requires the assumption: 
\begin{align}\label{eq:sc1}
    \min_{w \in \mathcal{W}_{k,\epsilon}} \left\| \sum_{i=1}^k w_{i}\left(x_{i}- \nu\right)\left(x_{i}- \nu\right)^{\top} \right\| \leq \lambda.
\end{align}
Intuitively,  this says that the data are roughly clustered around $\nu$ and no bad point significantly corrupts its shape. Note that by linearity, the objective can be rewritten as a minimax one, and this leads to the following definition
\begin{definition}[$(\epsilon, \lambda)$-spectral center]\label{def:spectral-center}
 A point $\nu \in \Real^d$ is a $(\epsilon, \lambda)$-spectral center of $\{x_i\}_{i=1}^k$ if  
\begin{align}\label{eq:one1}\tag{spectral center}
     \min_{w\in \mathcal{W}_{k,\epsilon}} \max_{M \succeq 0, \text{Tr}(M) = 1} \,\,\sum_{i=1}^k w_i \left\langle (x_i - \nu) (x_i - \nu)^{\top} , M\right \rangle \leq \lambda.
\end{align} 
\end{definition}
In the robust mean estimation setting, the deterministic conditions (\autoref{lem:second}) imply that the true mean is a $(\epsilon, O(1))$-spectral center.

\paragraph{Combinatorial center} 
On other hand, there is another natural way of saying that the data are centered around $\nu$, which proves to be more useful in the heavy-tailed setting. We call it \textit{combinatorial centrality} condition.  It roughly says that when we project the data onto \textit{any} one-dimensional direction,  a majority of them will be  close to $\nu$.
\begin{definition}[$(\epsilon, \lambda)$-combinatorial center]\label{def:comb-center}
A point $\nu$ is a $(\epsilon, \lambda)$-combinatorial center of $\{x_i\}_{i=1}^k$ if   for all unit $v \in \mathbb R^d$.
\begin{align}\label{eq:two}\tag{combinatorial center}
        \sum_{i=1}^{k} \I \left\{ \langle x_i - \nu, v \rangle \geq \sqrt \lambda\right\} \leq \epsilon k,
    \end{align}
\end{definition}

Lugosi and Mendelson \cite{lugosi2019sub} show that optimal confidence intervals for mean estimation can be obtained in the heavy-tailed model by finding combinatorial centers. (We elaborate in the next section.)
%In the heavy-tailed model, 
%the condition directly implies that $\nu$ has distance $2\sqrt{\lambda}$ close to the true mean by $\lambda$, as shown by Lugosi-Mendelson~\cite{lugosi2019sub}, and thus it suffices for the algorithm to find a combinatorial center (with the smallest possible $\lambda$ and a constant $\epsilon < 1/2$). 

\paragraph{Duality} It turns out that for constant $\epsilon$ these two conditions are equivalent (up to some minor gaps in constants).  To pave way for the proofs,  a key observation, first made by~\cite{cheng2019high}, is that the left-side of~\eqref{eq:one1} is an SDP objective. (This is because it is simply minimizing the maximum eigenvalue of $\sum_i w_i  (x_i - \nu) (x_i - \nu)^{\top} $.)
And strong duality allows us to swap the min and max, so 
\begin{align}\label{eq:sdpdual}
      \min_{w\in \mathcal{W}_{k,\epsilon}} \max_{M} \,\,\sum_{i=1}^k w_i \left\langle (x_i - \nu) (x_i - \nu)^{\top} , M\right \rangle =  \max_{M}  \min_{w\in \mathcal{W}_{k,\epsilon}} \,\,\sum_{i=1}^k w_i \left\langle (x_i - \nu) (x_i - \nu)^{\top} , M\right \rangle,
\end{align}
where the maximization is over the set of density matrices. Using this, we prove the following two propositions, showing (by contrapositives) that the two notions of centrality are equivalent. The constants   in the statements are chosen only to serve the purpose of heavy-tailed mean estimation, and they can be tweaked easily by the same arguments.

We consider the easy direction first. 
\begin{proposition}[spectral center$\implies$combinatorial center]\label{prop:p1}
If for some unit  $v\in  \mathbb R^d$
\begin{align}\label{eq:ass1}
    \sum_{i=1}^{k} \I \left\{| \langle x_i - \nu, v \rangle| \geq 10\sqrt \lambda\right\} \geq 0.4k,
\end{align}
then we have that for $\epsilon = 0.3$,
\begin{align*}
\min_{w\in \mathcal{W}_{k,\epsilon}} \max_{M \succeq 0, \text{Tr}(M) = 1} \,\,\sum_{i=1}^k w_i \left\langle (x_i - \nu) (x_i - \nu)^{\top} , M\right \rangle \geq \lambda.
\end{align*}
\end{proposition}
\begin{proof}
The assumption~\eqref{eq:ass1} immediately implies that 
\begin{align*}
        \sum_{i=1}^{k} \I \left\{ \langle x_i - \nu, v \rangle^2 \geq 100 \lambda\right\} \geq 0.4k
\end{align*}
This means that there are (at least) $0.4k$  points in $ \{x_i\}_{i=1}^k$ such that $t_i :=\left\langle (x_i - \nu) (x_i - \nu)^{\top} , M\right \rangle \geq 100\lambda$, where $M=vv^{\top}$.  We call them outliers.

Now by the SDP duality~\eqref{eq:sdpdual}, we only need to show  that for \textit{any} feasible $w$ the objective is at least $\lambda$.  Observe first that  for a fixed $M$, the optimal $w^*$ for the max-min objective is to put weight $1/(1-\epsilon)k$ on the $(1-\epsilon)k$ points with  the smallest $t_i$.
Recall we set $\epsilon = 0.3$. Hence, by pigeonhole principle, the support of $w^*$ must have an overlap of size $0.1k$ with the  outliers. It follows that 
\begin{align*}
   \sum_{i=1}^k w_i^* \left\langle (x_i - \nu) (x_i - \nu)^{\top} , M\right \rangle \geq 0.1k \cdot \frac{1}{(1-0.3)k} \cdot 100\lambda  \geq 10\lambda.  
\end{align*}
Since $w^*$ is the optimal choice, this completes the proof.
\end{proof}
The other direction is a bit more involved. The  key idea is to round the maximizing PSD matrix $M$ into a single vector $v$, via gaussian sampling, and this part of the argument is due to~\cite{lecue2019robust}.
\begin{proposition}[combinatorial center$\implies$spectral center]\label{prop:2}
Let $\epsilon = 0.1$. If for some $\nu \in \mathbb R^d$
\begin{align*}
\min_{w\in \mathcal{W}_{k,\epsilon}} \max_{M \succeq 0, \text{Tr}(M) = 1} \,\,\sum_{i=1}^k w_i \left\langle (x_i - \nu) (x_i - \nu)^{\top} , M\right \rangle \geq \lambda
\end{align*}
then we have for some unit $v$,
\begin{align*}
     \sum_{i=1}^{k} \I \left\{ |\langle x_i - \nu, v \rangle| \geq 0.1\sqrt \lambda\right\} \geq 0.01k.
\end{align*}
\end{proposition}
\begin{proof}
Strong duality~\eqref{eq:sdpdual} implies that there exists PSD $M$ of unit trace such that 
\begin{align*}
    \sum_{i=1}^k w_i \left\langle (x_i - \nu) (x_i - \nu)^{\top} , M\right \rangle \geq \lambda
\end{align*}
for all $w\in \mathcal{W}_{k,\epsilon}$. As we observed, the optimal $w^*$ for a fixed $M$ would put weights on the points with smallest value of $t_i= \left\langle (x_i - \nu) (x_i - \nu)^{\top} , M\right \rangle$. The fact that the objective is large implies that there must be more than $\epsilon k = 0.1k$ points with $t_i \geq \lambda$. Let $B$ be this set of points such that $t_i \geq\lambda$.

It remains to demonstrate a vector $v$ such that 
\begin{align}\label{eqn:badcenter}
     \sum_{i=1}^{k} \I \left\{ |\langle x_i - \nu, v \rangle| \geq 0.1\sqrt \lambda\right\} \geq 0.01k.
\end{align}
The idea is to round the PSD matrix $M$ to a single vector $v$ that achieves this inequality. The right   rounding method is simply \textit{gaussian sampling}. Namely, if we draw $v_M\sim \mathcal N(0, M)$, then it can be shown   that with constant probability  $v=v_M/\|v_M\|$ satisfies the property above. 

For that, we apply the argument from~\cite{lecue2019robust}. First let $g_i =\langle   x_i -\nu, v_M \rangle$ for each $i\in [k]$. Note that $g_i$ is a mean-zero Gaussian random variable with variance $\sigma_i^2= t_i$.  A standard anti-concentration calculation shows that for any $i \in B$, $\Pr ( |g_i| \geq 0.5\sqrt{\lambda}  ) \geq 1/2$. Therefore, if we define $$Y =\sum_{i=1}^{k} \I \left\{ |\langle x_i - \nu, v \rangle| \geq 0.5\sqrt \lambda\right\},$$ then by linearity of expectations we have $\E Y \geq 0.05k$. It follows from the Payley-Zigmund inequality that $\Pr (Y \geq 0.01k)\geq 0.0018$. Moreover, by Borell-TIS inequality (Theorem 7.1 of~\cite{ledoux2001concentration}), we can bound that with probability at least $0.999$, $$\|v_M\| \leq \E \|v_M\| + 4 \sqrt{\|M\|} \leq \sqrt{\Tr(M)} + 4\sqrt{\Tr(M)} \leq 5,$$
since $\Tr(M) = 1$. Combining these facts immediately proves~\eqref{eqn:badcenter}.
\end{proof}
\section{Estimation under Heavy-Tails}\label{sec:heavy-tail}
We now come to the heavy-tailed mean estimation problem and show how to solve it using the machinery developed in the last sections. The setting is very simple
\begin{definition}[heavy-tailed mean estimation with optimal rates]\label{def:heavy-tail}
Given $n$   random vectors $\{X_i\}_{i=1}^n$   drawn i.i.d.\ from a distribution $D$ over $\mathbb R^d$ with mean $\mu$ and (finite) covariance $\Sigma$ and a desired confidence $2^{-O(n)}\leq \delta < 1$, compute an estimate $\widehat \mu$ such that  with probability at least $1-\delta$,
\begin{align}\label{eq:optimal-rate}
\|\widehat{{\mu}} - {\mu} \|  \lesssim  r_\delta \overset{\mathrm{def}}{=} \sqrt{\frac{\Tr (\Sigma)}{n}} + \sqrt{\frac{ \| \Sigma \| \log (1/ \delta)}{n}}.
\end{align}
\end{definition}
We note that the error rate~\eqref{eq:optimal-rate} is information-theoretically optimal, up to a constant.
The bound is known as \textit{sub-gaussian} error, since when $D$ is sub-gaussian, the empirical average obtains the guarantee. 
Moreover, in general, the estimator needs to depend on the parameter $\delta$, and the requirement that  $\delta \geq 2^{-O(n)}$ is necessary~\cite{catoni2012challenging, devroye2016sub}.
In the following, we will  aim only at a  computationally efficient, $\delta$-\textit{dependent} construction  that attains the optimal error $r_\delta$. 

\paragraph{Lugosi-Mendelson Estimator.}  
In one dimension,  the well-known \textit{median-of-means} construction, due to~\cite{nemirovsky1983median, jerrum1986median, alon1999median}, provides such strong guarantee:
\begin{enumerate}[(i)]
    \item Bucket the data into $k = \lceil 8\log (1/\delta)\rceil$ disjoint groups and compute their means $Z_1,Z_2,\cdots, Z_k$.
    \item Output the median $\widehat{\mu}$ of $\{ Z_1, Z_2,\cdots, Z_k\}$.
\end{enumerate}
In high dimensions, however, the question is a lot more subtle, with the correct notion of median being elusive. A long line of work culminated in the celebrated work of Lugosi and Mendelson~\cite{lugosi2019sub}. The estimator follows the median-of-means paradigm by first bucketing the data into $k$ groups and taking the means $\{Z_i\}_{i=1}^k$.  The key structural lemma of their work is that the true mean is a $\left(0.01,O\left(r^2_\delta\right)\right)$-combinatorial center of the bucket means, where $r_\delta$ is the sub-gaussian error rate~\eqref{eq:optimal-rate}.  Recall that it means that if  we consider   projecting the bucket means to a one-dimensional direction,  a majority of them are close to the true mean. 
\begin{lemma}[Lugosi-Mendelson structural lemma~\cite{lugosi2019sub}]\label{lem:lm}
Consider the setting of heavy-tailed mean estimation (\autoref{def:heavy-tail}). Let $\{Z_i\}_{i=1}^k$ be the $k$ bucket means with  $k = \lceil 800\log (1/\delta)\rceil$.
Then with probability at least $1-\delta$,  \textit{for all} unit $v\in \mathbb R^d$, 
\begin{align}\label{eq:clus}\tag{$E_v$}
    | \langle Z_i - \mu , v\rangle  | \leq 3000 \left(\sqrt{\frac{\Tr (\Sigma)}{n}} + \sqrt{\frac{ \| \Sigma \| \log (1/ \delta)}{n}}\right),
\end{align}
for  $0.99k$ of the bucket means $\{Z_i\}_{i=1}^k$. 
\end{lemma}
This is exactly the combinatorial centrality condition (\autoref{def:comb-center}) we introduced in Section~\ref{sec:centrality}. To build more intuition, we should visualize it as a clustering property. That is, under any one-dimensional projection, the bucket means are clustered around the true mean, and the width of the cluster is precisely the optimal sub-gaussian error $O(r_\delta)$. 

This enables a natural estimator/algorithm---we can search for a point $\widehat \mu$ that is a $\left(0.01,r_\delta^2\right)$-combinatorial center for $\{Z_i\}_{i=1}^k$. Of course, such $\widehat\mu$ exists, since Lugosi-Mendelson (\autoref{lem:lm}) showed  that $\mu$ itself satisfies the condition (with probability at least $1-\delta$).  Furthermore, one can check any valid $(\epsilon, O(r_\delta^2))$-combinatorial center (\autoref{def:comb-center}) $\widehat\mu$ with $\epsilon<1/2$ is indeed   an estimator with sub-gaussian error rate $O(r_\delta)$, by a simple ``pigeonhole + triangle inequality'' argument. 
\begin{lemma}[combinatorial center has  sub-gaussian rate]\label{lem:center-error}
Let $\{Z_i\}_{i=1}^k$ be defined as above and $\epsilon < 1/2$. Suppose that the condition in the Lugosi-Mendelson structural lemma (\autoref{lem:lm}) holds. Then any $\left(\epsilon,O\left(r_\delta^2\right)\right)$-combinatorial center $\widehat \mu$ of $\{Z_i\}_{i=1}^k$  attains the sub-gaussian error~\eqref{eq:optimal-rate} (up to constant).
\end{lemma}
\begin{proof}
Let $v$ be the unit vector in the direction of $\mu - \widehat \mu$. Then since $\widehat \mu$ is an $(\epsilon, O(r_\delta^2))$-combinatorial center with $\epsilon <1/2$, we have $|  \langle Z_i - \widehat\mu, v \rangle |\le r_\delta$
for most $Z_i$. Also, $|  \langle Z_i - \mu, v \rangle |\le O(r_\delta)$ for most $\{Z_i\}_{i=1}^k$ by our assumption from Lugosi-Mendelson lemma. By the pigeonhole principle, there must be a $Z_j$ such that $| \langle Z_j - \widehat\mu, v \rangle |\le O(r_\delta)$ and $| \langle Z_j - \mu, v \rangle |\le O(r_\delta)$. By triangle inequality,
\begin{align*}
    \| \widehat\mu -\mu\| = \langle \mu - \widehat \mu ,v\rangle  \leq | \langle Z_i - \mu, v \rangle |+| \langle Z_i - \widehat\mu, v \rangle |  \le O(r_\delta).
\end{align*}
as desired, and this completes the proof.
\end{proof}
However, the problem of efficiently finding a combinatorial center appears difficult. If one sticks to its definition, it is required to  ensure that for \textit{all} unit vector $v$, the clustering property~\eqref{eq:clus} holds. It seems that even just \textit{certifying} this condition would na\"ively take exponential time (say, by enumerating a $1/2$-net of unit sphere). 
Yet, we can actually resort to duality, to avoid the pain of designing a new algorithm from scratch. As we showed,  a combinatorial center is just  a spectral center, which  our meta-algorithm can find for us.

\begin{theorem}[heavy-tailed mean estimation via spectral sample reweighing]\label{thm:reweight-heavy}
Given  $\{X_i\}_{i=1}^n$  and $\delta$, any constant-factor approximation  algorithm for  the spectral sample reweighing problem (\autoref{def:ssr})  can be used to compute an estimate $\widehat \mu$  that obtains the sub-gaussian error rate for heavy-tailed mean estimation (\autoref{def:heavy-tail}), with probability at least $1-\delta$.
\end{theorem}
\begin{proof}
Let $\{Z_i\}_{i=1}^k$ be the bucket means with $k=  \lceil 800\log (1/\delta)\rceil$ and let $\lambda = 3000 r_\delta$.
We assume  that the true mean $\mu$ is a $(0.01, \lambda^2)$-combinatorial center  of $\{Z_i\}_{i=1}^k$.
%, which holds with probability  at least $1-\delta$ by~\autoref{lem:lm}.
Suppose that we can obtain an $\alpha$-factor approximation the spectral sample reweighing, with the input being $\{Z_i\}_{i=1}^k$.
\begin{itemize}
    \item \textit{Promise}: First let's check the spectral centrality condition holds.  
    Since, by assumption, $\mu$ is a $(0.01, \lambda^2)$-combinatorial center  of $\{Z_i\}_{i=1}^k$, we have that  for all unit $v$
\begin{align*}
     \sum_{i=1}^{k} \I \left\{ |\langle x_i - \mu, v \rangle| \geq  \lambda \right\} \le 0.01k.
\end{align*}
Thus,~\autoref{prop:2} (with $\nu = \mu$)  implies that 
\begin{align*}
\min_{w\in \mathcal{W}_{k,\epsilon}} \max_{M \succeq 0, \text{Tr}(M) = 1} \,\,\sum_{i=1}^k w_i \left\langle (x_i - \mu) (x_i - \mu)^T , M\right \rangle \leq 100\lambda^2,
\end{align*}
where $\epsilon = 0.1$.
This means that there exists $w \in \mathcal{W}_{k, \epsilon}$ such that
\begin{equation*}\label{eq:pp}  
    \left\| \sum_{i=1}^n w_{i}\left(x_{i}- \mu\right)\left(x_{i}- \mu\right)^{T} \right\|\leq  100\lambda^2 .
\end{equation*} 
\item \textit{Output}:
Now the guarantee of an $\alpha$-factor approximation for spectral sample reweighing (\autoref{def:ssr})  is that we have $\widehat\mu \in \mathbb{R}^d$ and $w'\in\mathcal{W}_{k,3\epsilon}$ such that 
\begin{equation*}\label{eq:pp2}  
    \left\| \sum_{i=1}^n w'_{i}\left(x_{i}- \widehat \mu\right)\left(x_{i}-\widehat \mu\right)^{T} \right\|\leq   100\alpha \lambda^2.
\end{equation*} 
It immediately follows that
\begin{align*}
\min_{w\in \mathcal{W}_{k,3\epsilon}} \max_{M \succeq 0, \text{Tr}(M) = 1} \,\,\sum_{i=1}^k w_i \left\langle (x_i - \widehat\mu) (x_i - \widehat\mu)^T , M\right \rangle \leq 100\alpha\lambda^2.
\end{align*}
Now we can apply~\autoref{prop:p1}. Since $\alpha$ is a constant by assumption, we obtain that for all unit $v$,
\begin{align}
     \sum_{i=1}^{k} \I \left\{ \left|\left\langle x_i - \widehat\mu, v \right\rangle\right| \geq   C(\alpha)\cdot\lambda \right\} \le 0.4k,
\end{align}
for some constant $C(\alpha)=O(1)$ that depends on $\alpha$.
Therefore, we get that a majority of the points cluster around $\widehat \mu$, along any direction $v$, so it is a $(0.4,O(\lambda))$-combinatorial center.  It follows from  \autoref{lem:center-error} that $\|\widehat{\mu}-\mu \| \leq O(r_\delta)$, as $\lambda = O(r^2_\delta)$. 
\end{itemize}
Finally, note that the only condition of the argument is that the true mean is a combinatorial center, which occurs with probability at least $1-\delta$, by~\autoref{lem:lm}.
\end{proof}
We remark that the exact constants we choose in the proof are immaterial, and no efforts have been given in optimizing them.

The theorem   implies that the filter algorithm (\autoref{alg:filter}) combined with a simple pruning step from \cite{lei2019fast}) can be used for heavy-tailed mean estimation as well. 
\begin{corollary}[filter for heavy-tailed mean estimation]\label{cor:heavy-filter}
Given $\{X_i\}_{i=1}^n$ drawn i.i.d.\ from a distribution with mean $\mu$ and covariance $\Sigma$ and a failure probability $2^{-O(n)}\leq \delta <1$, there is an efficient algorithm that outputs $\widehat \mu$ such that  with probability at least $1-\delta$, $\|\widehat{{\mu}} - {\mu} \|\le O(r_\delta$).

Further, the algorithm is a black-box application of \autoref{alg:filter} and runs in time $O(k^2 d^2 + nd) $.
\end{corollary}
\begin{proof}
Given the input, we first compute the bucket means  $\{Z_i\}_{i=1}^{2k}$, which takes $O(nd)$ time. Assume that the condition of the Lugosi-Mendelson structural lemma (\autoref{lem:lm}) holds; that is, $\mu$ is a $(0.01, \lambda^2)$-combinatorial center  of $\{Z_i\}_{i=1}^k$, where $\lambda = 3000r_\delta$. We use the   filter algorithm (Algorithm \ref{alg:filter}) with the input being a pruned subset of  $\{Z_i\}_{i=1}^k$ and apply its guarantees.

Here, we will not use the pruning step (\autoref{lem:naive-prune}), since it requires the knowledge of $\lambda$. Instead, we first compute the coordinate-wise median-of-means $\widehat\mu_0$ of $\{Z_i\}_{i=k+1}^{2k}$ and the distances $d_i= \| Z_i - \widehat\mu_0\|$ for each $i\in [k]$. We then sort the points by $d_i$ (in descending order)  and remove the top $0.01k$ points in $\{Z_i\}_{i=1}^k$ with large $d_i$.
It can be shown that the remaining points has diameter at most $O(\sqrt{d}r_\delta)$; see Lemma E.1 of~\cite{lei2019fast}.  Let  $S$ the remaining points  in $\{Z_i\}_{i=1}^k$.
% We now give a proof sketch. 
% Observe that by the  the condition of the Lugosi-Mendelson lemma, $\|\widehat\mu_0-\mu \| \leq O(\sqrt{d} r_\delta)$. Moreover, since each $Z_i$ is an average of $\lfloor n/2k \rfloor$ samples, Markov's inequality implies that $\|Z_i - \mu\| \le O(\sqrt{d} r_\delta)$ with probability $0.999$. By triangle inequality, for each $Z_i$, it holds that $\|Z_i -\widehat \mu_0 \| \le O(\sqrt{d}r_\delta)$, with probability $0.999$.  By a standard binomial tail bound,  $\|Z_i - \mu\| \ge \Omega (\sqrt{d} r_\delta)$ for at most  $0.1k$ indices in $[k]$, with probability at least $1-\delta/3$ 

For the run-time, we can   apply the   guarantee of the filter algorithm (\autoref{lem:filter}), given the input $S$ and a  failure probability   $\delta/3$. Since the squared diameter is $\rho = O(dr_\delta^2)$ and $\lambda = O(r_\delta^2)$, this gives a run-time of $\widetilde O(k^2d^2)$, since $k = O(\log (1/\delta))$.

%When $k\geq d$, we can directly apply the   guarantee of the filter algorithm (\autoref{thm:filter}), where we set its failure probability to be at most $\delta/3$.  This gives an algorithm with run-time $O(k^3d)$, since $k = O(\log (1/\delta))$. If $k< d$, we first project the data onto the subspace containing  $\{Z_i\}_{i=1}^k$, which takes $O(k^2d)$ time, so we can assume their dimension is $k$.\footnote{This trick is due to~\cite{cherapanamjeri2019fast}.}  This means that each application of power method in \autoref{alg:filter} runs in time $\widetilde O(k^3)$,   since  $k = O(\log (1/\delta))$ and we set the failure probability of each call to power method as $\delta/T$, where the number of iterations is $T=O(d)$ as before. 

We now have a constant-factor approximation for the spectral sample reweighing problem.  By~\autoref{thm:reweight-heavy}, this gives an estimate with the sub-gaussian error~\eqref{eq:optimal-rate}.
Finally, the procedure's success depends on the   condition of Lugosi-Mendelson (\autoref{thm:reweight-heavy}), success of the pruning procedure,  and the guarantees of constant-approximation of spectral sample reweighing (\autoref{thm:filter}). The failure probability of each event can be bounded by $\delta/3$. Applying union bound completes the proof.
\end{proof}

\paragraph{Other algorithms for heavy-tailed mean estimation}
This argument also enables us to solve the heavy-tailed mean estimation problem using other approximation algorithms for the spectral sample reweighing problem. 
Let $\lambda = 3000r_\delta$.
Recall that the argument for \autoref{thm:reweight-heavy} shows that there is a $(0.1, O(\lambda^2))$-spectral center (which is the true mean $\mu$). 
Moreover,   the pruning step in the proof of \autoref{cor:heavy-filter} allows us to bound the squared diameter of a large subset of $\{Z_i\}_{i=1}^k$ by  $\rho = O(d\lambda^2)$.

This implies that the gradient descent-based algorithm that   we analyze in Appendix~\ref{sec:gd} solves the heavy-tailed setting in $O\left(kd^2\right)$ iterations.
\begin{corollary}[heavy-tailed mean estimation via gradient descent]\label{cor:heavy-gd}
Assume the setting of \autoref{cor:heavy-filter}. A black-box application of the gradient descent-based algorithm (\autoref{alg:gd}, Appendix~\ref{sec:gd}) solves the heavy-tailed mean estimation problem with optimal error rate within $O(n d^2)$ iterations and $\widetilde O(n^2 d^3)$ time.
\end{corollary}

The quantum entropy scoring scheme (Appendix~\ref{sec:mmwu}), however, runs in $\widetilde O(\log (\rho/\lambda))$ number of iterations. Setting its failure probability to be $\delta/3$, we obtain the following, which matches the fastest-known algorithm for the problem~\cite{lecue2019robust,lei2019fast}.
\begin{corollary}[heavy-tailed mean estimation via quantum entropy scoring]
Assume the setting of \autoref{cor:heavy-filter}. A black-box application of the matrix multiplicative update algorithm (\autoref{alg:mmwu}, Appendix~\ref{sec:mmwu}) solves the heavy-tailed mean estimation problem with optimal error rate, in   $\widetilde O(1)$ iterations and   $\widetilde O(k^2d)$ total run-time.
\end{corollary}

\section{Discussion}
Estimating the mean of a distribution is arguably the most fundamental problem in statistics.   We showed that in robust and heavy-tailed settings, the problem  can be approached by techniques from regret minimization and online learning.  We believe the ideas we present here may be more broadly applicable to other  problems in high-dimensional robust statistics, such regression and covariance estimation. 

\section*{Acknowledgments} 
F.\ Z.\ would like to thank  Banghua Zhu for helpful discussions and    Prayaag Venkat for comments on  an early draft of this work.
\bibliographystyle{alpha}
\bibliography{bib}

\newcommand{\etalchar}[1]{$^{#1}$}
\begin{thebibliography}{DKK{\etalchar{+}}19b}

\bibitem[AHK12]{arora2012multiplicative}
Sanjeev Arora, Elad Hazan, and Satyen Kale.
\newblock The multiplicative weights update method: a meta-algorithm and
  applications.
\newblock {\em Theory of Computing}, 8(1):121--164, 2012.

\bibitem[Alb05]{albert2005scale}
Reka Albert.
\newblock Scale-free networks in cell biology.
\newblock {\em Journal of Cell Science}, 118(21):4947--4957, 2005.

\bibitem[AMS99]{alon1999median}
Noga Alon, Yossi Matias, and Mario Szegedy.
\newblock The space complexity of approximating the frequency moments.
\newblock {\em J. Comput. System Sci.}, 58(1):137--147, 1999.

\bibitem[AZLO15]{allen2015spectral}
Zeyuan Allen-Zhu, Zhenyu Liao, and Lorenzo Orecchia.
\newblock Spectral sparsification and regret minimization beyond matrix
  multiplicative updates.
\newblock In {\em ACM Symposium on Theory of Computing (STOC '15)}, 2015.

\bibitem[Bar05]{barabasi2005origin}
Albert-Laszlo Barabasi.
\newblock The origin of bursts and heavy tails in human dynamics.
\newblock {\em Nature}, 435(7039):207--211, 2005.

\bibitem[BDLS17]{balakrishnan2017computationally}
Sivaraman Balakrishnan, Simon~S Du, Jerry Li, and Aarti Singh.
\newblock Computationally efficient robust sparse estimation in high
  dimensions.
\newblock In {\em Conference on Learning Theory (COLT '17)}, 2017.

\bibitem[Cat12]{catoni2012challenging}
Olivier Catoni.
\newblock Challenging the empirical mean and empirical variance: a deviation
  study.
\newblock {\em Annales de l'IHP Probabilit{\'e}s et statistiques},
  48(4):1148--1185, 2012.

\bibitem[CDG19]{cheng2019high}
Yu~Cheng, Ilias Diakonikolas, and Rong Ge.
\newblock High-dimensional robust mean estimation in nearly-linear time.
\newblock In {\em ACM-SIAM Symposium on Discrete Algorithms (SODA '19)}, 2019.

\bibitem[CDGS20]{cheng2020high}
Yu~Cheng, Ilias Diakonikolas, Rong Ge, and Mahdi Soltanolkotabi.
\newblock High-dimensional robust mean estimation via gradient descent.
\newblock In {\em International Conference on Machine Learning (ICML '20)},
  2020.

\bibitem[CDGW19]{cheng2019faster}
Yu~Cheng, Ilias Diakonikolas, Rong Ge, and David~P Woodruff.
\newblock Faster algorithms for high-dimensional robust covariance estimation.
\newblock In {\em Conference on Learning Theory (COLT '19)}, 2019.

\bibitem[CFB19]{cherapanamjeri2019fast}
Yeshwanth Cherapanamjeri, Nicolas Flammarion, and Peter~L Bartlett.
\newblock Fast mean estimation with sub-gaussian rates.
\newblock In {\em Conference on Learning Theory (COLT '19)}, 2019.

\bibitem[DHL19]{dong2019quantum}
Yihe Dong, Samuel~B Hopkins, and Jerry Li.
\newblock Quantum entropy scoring for fast robust mean estimation and improved
  outlier detection.
\newblock In {\em Neural Information Processing Systems (NeurIPS '19)}, 2019.

\bibitem[DK19]{diakonikolas2019recent}
Ilias Diakonikolas and Daniel~M Kane.
\newblock Recent advances in algorithmic high-dimensional robust statistics.
\newblock {\em arXiv preprint arXiv:1911.05911}, 2019.

\bibitem[DKK{\etalchar{+}}17]{diakonikolas2017being}
Ilias Diakonikolas, Gautam Kamath, Daniel~M Kane, Jerry Li, Ankur Moitra, and
  Alistair Stewart.
\newblock Being robust (in high dimensions) can be practical.
\newblock In {\em International Conference on Machine Learning (ICML '17)},
  2017.

\bibitem[DKK{\etalchar{+}}18]{diakonikolas2018robustly}
Ilias Diakonikolas, Gautam Kamath, Daniel~M Kane, Jerry Li, Ankur Moitra, and
  Alistair Stewart.
\newblock Robustly learning a gaussian: Getting optimal error, efficiently.
\newblock In {\em ACM-SIAM Symposium on Discrete Algorithms (SODA '18)}. SIAM,
  2018.

\bibitem[DKK{\etalchar{+}}19a]{diakonikolas2019robust}
Ilias Diakonikolas, Gautam Kamath, Daniel Kane, Jerry Li, Ankur Moitra, and
  Alistair Stewart.
\newblock Robust estimators in high-dimensions without the computational
  intractability.
\newblock {\em SIAM Journal on Computing}, 48(2):742--864, 2019.

\bibitem[DKK{\etalchar{+}}19b]{diakonikolas2019sever}
Ilias Diakonikolas, Gautam Kamath, Daniel Kane, Jerry Li, Jacob Steinhardt, and
  Alistair Stewart.
\newblock Sever: A robust meta-algorithm for stochastic optimization.
\newblock In {\em International Conference on Machine Learning (ICML '19)},
  2019.

\bibitem[DKP20]{diakonikolas2020outlier}
Ilias Diakonikolas, Daniel~M Kane, and Ankit Pensia.
\newblock Outlier robust mean estimation with subgaussian rates via stability.
\newblock {\em arXiv preprint arXiv:2007.15618}, 2020.

\bibitem[DKS18]{diakonikolas2018list}
Ilias Diakonikolas, Daniel~M Kane, and Alistair Stewart.
\newblock List-decodable robust mean estimation and learning mixtures of
  spherical gaussians.
\newblock In {\em ACM Symposium on Theory of Computing (STOC '18)}, 2018.

\bibitem[DKS19]{diakonikolas2019efficient}
Ilias Diakonikolas, Weihao Kong, and Alistair Stewart.
\newblock Efficient algorithms and lower bounds for robust linear regression.
\newblock In {\em ACM-SIAM Symposium on Discrete Algorithms (SODA '19)}, 2019.

\bibitem[DL19]{lecue2019robust}
Jules Depersin and Guillaume Lecu{\'e}.
\newblock Robust subgaussian estimation of a mean vector in nearly linear time.
\newblock {\em arXiv:1906.03058}, 2019.

\bibitem[DLLO16]{devroye2016sub}
Luc Devroye, Matthieu Lerasle, Gabor Lugosi, and Roberto~I Oliveira.
\newblock Sub-gaussian mean estimators.
\newblock {\em Annals of Statistics}, 44(6):2695--2725, 2016.

\bibitem[FFF99]{faloutsos1999power}
Michalis Faloutsos, Petros Faloutsos, and Christos Faloutsos.
\newblock On power-law relationships of the internet topology.
\newblock {\em ACM SIGCOMM Computer Communication Review}, 29(4):251--262,
  1999.

\bibitem[Haz16]{hazan2016introduction}
Elad Hazan.
\newblock Introduction to online convex optimization.
\newblock {\em Foundations and Trends in Optimization}, 2(3-4):157--325, 2016.

\bibitem[HL18]{hopkins2018mixture}
Samuel~B Hopkins and Jerry Li.
\newblock Mixture models, robustness, and sum of squares proofs.
\newblock In {\em ACM SIGACT Symposium on Theory of Computing (STOC '18)},
  2018.

\bibitem[Hop20]{hopkins2018mean}
Samuel~B. Hopkins.
\newblock Mean estimation with sub-gaussian rates in polynomial time.
\newblock {\em Annals of Statistics}, 48(2):1193--1213, 04 2020.

\bibitem[Hub64]{huber1964robust}
Peter~J Huber.
\newblock Robust estimation of a location parameter.
\newblock {\em The Annals of Mathematical Statistics}, 35(1):73--101, 1964.

\bibitem[HW01]{herbster2001tracking}
Mark Herbster and Manfred~K Warmuth.
\newblock Tracking the best linear predictor.
\newblock {\em Journal of Machine Learning Research}, 1(Sep):281--309, 2001.

\bibitem[JVV86]{jerrum1986median}
Mark~R. Jerrum, Leslie~G. Valiant, and Vijay~V. Vazirani.
\newblock Random generation of combinatorial structures from a uniform
  distribution.
\newblock {\em Theoret. Comput. Sci.}, 43(2-3):169--188, 1986.

\bibitem[KW92]{kuczynski1992estimating}
Jacek Kuczy{\'n}ski and Henryk Wo{\'z}niakowski.
\newblock Estimating the largest eigenvalue by the power and lanczos algorithms
  with a random start.
\newblock {\em SIAM Journal on Matrix Analysis and Applications},
  13(4):1094--1122, 1992.

\bibitem[Led01]{ledoux2001concentration}
Michel Ledoux.
\newblock {\em The concentration of measure phenomenon}.
\newblock American Mathematical Society, 2001.

\bibitem[Li18]{li2018principled}
Jerry~Zheng Li.
\newblock {\em Principled approaches to robust machine learning and beyond}.
\newblock PhD thesis, Massachusetts Institute of Technology, 2018.

\bibitem[Li19a]{jerrynote}
Jerry Li.
\newblock Lecture 4: Spectral signatures and efficient certifiability.
\newblock \url{https://jerryzli.github.io/robust-ml-fall19/lec4.pdf}, 2019.

\bibitem[Li19b]{jerrynote2}
Jerry Li.
\newblock Lecture 5: Filtering from spectral signatures.
\newblock \url{https://jerryzli.github.io/robust-ml-fall19/lec5.pdf}, 2019.

\bibitem[LKF05]{leskovec2005graphs}
Jure Leskovec, Jon Kleinberg, and Christos Faloutsos.
\newblock Graphs over time: densification laws, shrinking diameters and
  possible explanations.
\newblock In {\em ACM SIGKDD International Conference on Knowledge Discovery in
  Data Mining (KDD '05)}, 2005.

\bibitem[LLVZ20]{lei2019fast}
Zhixian Lei, Kyle Luh, Prayaag Venkat, and Fred Zhang.
\newblock A fast spectral algorithm for mean estimation with sub-gaussian
  rates.
\newblock In {\em Conference on Learning Theory (COLT '20)}, 2020.

\bibitem[LM19]{lugosi2019sub}
G{\'a}bor Lugosi and Shahar Mendelson.
\newblock Sub-gaussian estimators of the mean of a random vector.
\newblock {\em Annals of Statistics}, 47(2):783--794, 2019.

\bibitem[LM20]{lugosi2019robust}
Gabor Lugosi and Shahar Mendelson.
\newblock Robust multivariate mean estimation: the optimality of trimmed mean.
\newblock {\em Annals of Statistics}, 2020.

\bibitem[LRV16]{lai2016agnostic}
Kevin~A Lai, Anup~B Rao, and Santosh Vempala.
\newblock Agnostic estimation of mean and covariance.
\newblock In {\em IEEE Symposium on Foundations of Computer Science (FOCS
  '16)}. IEEE, 2016.

\bibitem[NY83]{nemirovsky1983median}
A.~S. Nemirovsky and D.~B. Yudin.
\newblock {\em Problem complexity and method efficiency in optimization}.
\newblock A Wiley-Interscience Publication. John Wiley \& Sons, Inc., New York,
  1983.
\newblock Translated from the Russian and with a preface by E. R. Dawson,
  Wiley-Interscience Series in Discrete Mathematics.

\bibitem[PBR19]{prasad2019unified}
Adarsh Prasad, Sivaraman Balakrishnan, and Pradeep Ravikumar.
\newblock A unified approach to robust mean estimation.
\newblock {\em arXiv preprint arXiv:1907.00927}, 2019.

\bibitem[SCV18]{steinhardt2018resilience}
Jacob Steinhardt, Moses Charikar, and Gregory Valiant.
\newblock Resilience: A criterion for learning in the presence of arbitrary
  outliers.
\newblock In {\em Innovations in Theoretical Computer Science Conference (ITCS
  '18)}, 2018.

\bibitem[SL14]{pmlr-v32-steinhardtb14}
Jacob Steinhardt and Percy Liang.
\newblock Adaptivity and optimism: An improved exponentiated gradient
  algorithm.
\newblock In {\em International Conference on Machine Learning (ICML '14)},
  2014.

\bibitem[Ste18]{steinhardt2018robust}
Jacob Steinhardt.
\newblock {\em Robust Learning: Information Theory and Algorithms}.
\newblock PhD thesis, Stanford University, 2018.

\bibitem[Tuk60]{tukey60}
John~W. Tukey.
\newblock A survey of sampling from contaminated distributions.
\newblock {\em Contributions to probability and statistics}, 2:448--485, 1960.

\bibitem[WK08]{warmuth2008randomized}
Manfred~K Warmuth and Dima Kuzmin.
\newblock Randomized online pca algorithms with regret bounds that are
  logarithmic in the dimension.
\newblock {\em Journal of Machine Learning Research}, 9(Oct):2287--2320, 2008.

\bibitem[WL15]{wang2015projection}
Weiran Wang and Canyi Lu.
\newblock Projection onto the capped simplex.
\newblock {\em arXiv preprint arXiv:1503.01002}, 2015.

\bibitem[Zin03]{zinkevich2003online}
Martin Zinkevich.
\newblock Online convex programming and generalized infinitesimal gradient
  ascent.
\newblock In {\em International Conference on Machine Learning (ICML '03)},
  2003.

\bibitem[ZJS20]{zhu2020robust}
Banghua Zhu, Jiantao Jiao, and Jacob Steinhardt.
\newblock Robust estimation via generalized quasi-gradients.
\newblock {\em arXiv preprint arXiv:2005.14073}, 2020.

\end{thebibliography}
\newpage 

\appendix
\section{Technical Lemmas and Proofs}

\subsection{Spectral Signatures}

\begin{lemma}\label{lem:ssw}
Let  $\{x_i\}_{i=1}^n$ be $n$ points in $\mathbb R^d$.
Suppose there exists $\nu\in \mathbb R^d$ and a set of  good weights $w \in \mathcal{W}_{n
    ,\epsilon}$ such that 
\begin{equation}\label{eq:p3}  
    \sum_{i=1}^n w_{i}\left(x_{i}- \nu\right)\left(x_{i}- \nu\right)^{\top}\preceq \lambda I. 
\end{equation}
for some   $\lambda  > 0$. Then  for any $w'\in \mathcal{W}_{n,\epsilon}$, 
    \begin{equation}\label{eq:ss2}
        \|\nu- \nu(w')\| \le  \frac{1}{1-\sqrt{2\epsilon}} \left(\sqrt{\lambda}+ \sqrt {2\epsilon\lambda} +   \sqrt{2\epsilon\|M(w')\|}\right),
    \end{equation}
    where $\nu(w') = \sum_{i}w'_i x_i$ and $M(w') = \sum_i w'_i (x_i - \nu(w'))(x_i - \nu(w'))^\top$.
\end{lemma}
The lemma and its proof strategy is similar to the spectral signature lemma in robust statistics and is now somewhat  standard in the literature; see, \eg, \cite{dong2019quantum,li2018principled}.

\begin{proof}
To bound $\|\nu- \nu(w')\|$, we note 
\begin{align}
    \|\nu(w') - \nu\|^2_2 &= \sum_i w_i' \, \left\langle \nu(w') - \nu, x_i- \nu\right\rangle.
    \end{align}
    In bounding this sum, we may assume without loss of generality that $w_i' > 0$ for all $i$. Now observe that we can decompose the sum as 
    \begin{align}
  \sum_i w_i' \, \left\langle \nu(w') - \nu, x_i- \nu\right\rangle      &=\sum_{i} w_i \left\langle \nu(w') - \nu, x_i- \nu\right\rangle  + \sum_{i:w_i >w_i'} \left(w_i' -  w_i\right)\left\langle \nu(w') - \nu, x_i- \nu\right\rangle \nonumber \\&\hspace{1em}+ \sum_{i:w_i' >w_i} \left(w_i' -  w_i\right)\left\langle \nu(w') - \nu, x_i- \nu\right\rangle.
\end{align}
We bound the three terms respectively as follows. 
\begin{enumerate}[(i)]
    \item For the first term, by Cauchy-Schwarz,
    \begin{align*}
  \sum_{i} w_i \left\langle \nu(w') - \nu, x_i- \nu\right\rangle   =   \left\langle \nu(w') - \nu, \nu(w) - \nu\right\rangle  \leq \|\nu(w') - \nu\| \cdot \|\nu(w) - \nu\|.
    \end{align*}
    By Jensen's inequality and the spectral centrality assumption, we have for all unit $u$,
    \begin{align*}
     \langle\nu(w)- \nu,u\rangle^2 =   \left\langle   \sum_{i}w_i x_i-\nu, u\right\rangle^2 \leq  \sum_{i} w_i\langle x_i-\nu , u \rangle^2 \leq \lambda.
    \end{align*}
    Thus, $\|\nu(w) - \nu\| \leq \sqrt{\lambda}$.
    \item For the second term, let $\alpha_i  = w_i' - w_i$. Then  note that if $w_i'  < w_i$, 
    \begin{equation}\label{eq:ratio-1}
            \left|\frac{\alpha_i}{w_i}\right| = \left|\frac{w_i'}{w_i}-1\right|\leq 1    \end{equation}
 Hence,    
 \begin{align}
      \left( \sum_{i:w_i>w_i'} \alpha_i\left\langle \nu(w') - \nu, x_i- \nu\right\rangle\right)^2 &\leq \left(\sum_{i:w_i>w_i'}\frac{1}{w_i} \alpha_i^2\right) \cdot \sum_{i:w_i>w_i'} w_i \left\langle \nu(w') - \nu, x_i- \nu\right\rangle^2\label{eq:ii0}\\
      &\leq \left(\sum_{i:w_i>w_i'}  \frac{1}{w_i} \alpha_i^2\right) \cdot \lambda \cdot  \|\nu(w')-\nu\|^2_2 \label{eq:ii1}\\
      &\leq   \left(\sum_{i:w_i>w_i'}  |\alpha_i|\right) \cdot \lambda \cdot  \|\nu(w')-\nu\|^2_2\label{eq:ii2}\\
      &\leq 2\epsilon\lambda  \|\nu(w')-\nu\|^2_2
      \end{align}
    where \eqref{eq:ii0} follows from Cauchy-Schwarz, \eqref{eq:ii1} follows from the spectral centrality assumption, and \eqref{eq:ii2} follows from \eqref{eq:ratio-1}.
    % where~\eqref{eq:ii1} is by the covariance bound that $A_G \preceq \lambda I$, ~\eqref{eq:ii2} follows since $|(1-\epsilon)n\alpha_i| \leq 1$, and~\eqref{eq:ii3} since $\sum_{i=1}^n |\alpha_i| \leq \epsilon/(1-\epsilon)\leq 2\epsilon$.
    \item For the third term, again let $\alpha_i  = w_i' - w_i$. Similarly, if $w_i'  > w_i$, 
    \begin{equation}\label{eq:ratio-2}
            \left|\frac{\alpha_i}{w_i'}\right| = \left|\frac{w_i}{w_i'}-1\right|\leq 1    \end{equation}
 It follows that    
 \begin{align}
      \left( \sum_{i:w_i'>w_i} \alpha_i\left\langle \nu(w') - \nu, x_i- \nu\right\rangle\right)^2 &\leq \left(\sum_{i:w_i'>w_i}\frac{1}{w_i'} \alpha_i^2\right) \cdot \sum_{i:w_i'>w_i} w_i' \left\langle \nu(w') - \nu, x_i- \nu\right\rangle^2\label{eq:ii5}\\
      &\leq \left(\sum_{i:w_i'>w_i}  \frac{1}{w_i'} \alpha_i^2\right) \cdot \sum_{i:w_i'>w_i} w_i' \left\langle \nu(w') - \nu, x_i- \nu\right\rangle^2 \label{eq:ii6}\\
      &\leq   \left(\sum_{i:w_i'>w_i}  |\alpha_i|\right) \cdot\sum_{i:w_i'>w_i} w_i' \left\langle \nu(w') - \nu, x_i- \nu\right\rangle^2\label{eq:ii7}\\
      &\leq 2\epsilon\cdot \sum_{i} w_i' \left\langle \nu(w') - \nu, x_i- \nu\right\rangle^2 
      \end{align}
      For the sum, we have 
    \begin{align}
         \sum_{i} w_i' \left\langle \nu(w') - \nu, x_i- \nu\right\rangle^2 &= \sum_i w_i'\left\langle \nu(w') - \nu, x_i- \nu(w')\right\rangle^2 + \|\nu(w')-\nu\|^4\\
        &\le  \|M(w')\|_2 \cdot \|\nu(w') - \nu\|^2+ \|\nu(w')-\nu\|^4.
    \end{align}
%    \fred{The fourth order term is a bit annoying, which leads to the $ 1/(1-\sqrt {2\epsilon})$ leading coefficient in the final bound. It doesn't matter for our purpose. But I'd like to get rid of it.}
    % By Cauchy-Schwarz,
    % \begin{align*}
    %     \left(\sum_{i}  w_i' \, \left\langle \nu(w') - \nu, x_i- \nu(w')\right\rangle\right)^2 &\leq \left(\sum_{i} w_i'\right) \left(\sum_{i} w_i' \cdot \left\langle \nu(w') - \nu, x_i- \nu(w')\right\rangle^2\right)\\
    %     &\leq 2\epsilon\cdot \sum_{i=1}^n w_i'\left\langle \nu(w') - \nu, x_i- \nu(w')\right\rangle^2 \\
    %     &= 2\epsilon \cdot \left(\nu(w') - \nu\right)^\top M(w')\left(\nu(w') - \nu\right)\\
    %     &\leq2 \epsilon \cdot \|M(w')\|_2 \cdot \|\nu(w') - \nu\|^2
    % \end{align*}
\end{enumerate}
Putting everything together and rearranging finishes the proof.
\end{proof}

\begin{lemma}\label{lem:robustspec}
Let $\{x_i\}_{i=1}^n$ be $n$ points. Suppose there exists a subset $G\subset [n]$ of size $(1-\epsilon)$ such that $\frac{1}{\left|G\right|} \sum_{i  \in G}\left(x_{i}-{\mu_G}\right)\left(x_{i}-{\mu_G}\right)^{\top}\preceq \lambda I$ for some $\lambda >0$, where $\mu_G = \frac{1}{|G|} \sum_{i \in G} x_i$. Then for any $w\in \mathcal{W}_{n,\epsilon}$, 
\begin{equation}\label{eq:ssg}
        \|\mu_G - \mu(w)\| \le  \frac{1}{1-2\epsilon} \left(\sqrt {2\epsilon\lambda} +   \sqrt{\epsilon\|M(w)\|}\right).
    \end{equation}
\end{lemma}
\begin{proof}
    The proof follows from the same argument of~\autoref{lem:ssw}, with $\nu=\mu_G$. %Observe that the first term in~\eqref{eq:bbb1} becomes $\frac{1}{|G|}\sum_{i \in G}   \left\langle \mu(w) - \nu, x_i- \nu\right\rangle$, which equals $0$ when $\nu = \mu_G$, shaving the $\sqrt{\lambda}$ term in the final bound. 
\end{proof}

%%%%%%%%%%%%%%%%%%%%%%%%%%%%%%%%%%%%%%%%%
%%%%%%%%%%%%%%%%%%%%%%%%%%%%%%%%%%%%%%%%%
%%%%%%%%%%%%%%%%%%%%%%%%%%%%%%%%%%%%%%%%%
\subsection{A KL Divergence Bound}

\begin{lemma}\label{lem:kl}
Let $p\in \mathcal{W}_{n,\epsilon}$ and $q$ be the uniform distribution over $n$ points. Then
$\text{KL} (p || q) \leq 5 \epsilon$.
\end{lemma}
\begin{proof}
    The lemma follows from direct calculations. By definition of KL divergence,
    \begin{align*}
        \text{KL} (p || q) &= \sum_i p_i \log\frac{p_i }{q_i}\\
                             &= \sum_i p_i \log (np_i)\\
                             &\leq \sum_i\frac{1}{(1-\epsilon)n} \log\frac{1}{(1-\epsilon)}
                              \\
                             &=\frac{1}{1-\epsilon}\log \frac{1}{1-\epsilon}  \\
                             &\leq 5\epsilon.
    \end{align*}
    where the last inequality holds when $0<\epsilon \leq 1/2$.
\end{proof}

\subsection{Proof of~\autoref{lem:diameter}}\label{sec:diameter}
\begin{proof}[Proof of~\autoref{lem:diameter}]
We show that there exists a ball of radius $\sqrt{d\lambda/\epsilon}$ that contains at least $(1-3\epsilon)n$ points. Note that the spectral centrality condition  $
    \sum_{i=1}^n w_{i}\left(x_{i}- \nu\right)\left(x_{i}- \nu\right)^{\top}\preceq \lambda I
    $ implies that 
\begin{align*}
    \sum_{i=1}^n \Tr \left(w_i (x_i -\nu)(x_i -\nu)^\top\right) \leq d\lambda.
\end{align*}
By the cyclic property of trace, we get
\begin{align*}
    \sum_{i=1}^n w_i \|x_i - \nu\|^2 \leq d\lambda.
\end{align*}
Therefore, by Markov's inequality, 
\begin{align}\label{eqn:markov}
    \Pr_{i\sim w}\left( \|x_i-\nu\|^2 \geq  d\lambda/\epsilon \right) \leq {\epsilon},
\end{align}
where $i\sim w$ denotes $i$ drawn from the discrete distribution defined by $w$. Observe that since $\mathcal{W}_{n,\epsilon}$ is the convex hull of all uniform distributions over a subset of size $(1-\epsilon)n$, we have $
    \left\|w- \mathcal{U}_n\right\|_1 \leq 2\epsilon$. Thus, $\text{TV}(w, \mathcal{U}_n) \leq \epsilon$. Hence, using the definition of total variation distance, \eqref{eqn:markov} implies that 
\begin{align}\label{eqn:markov-u}
    \Pr_{i\sim \mathcal{U}_n}\left( \|x_i-\nu\|^2 \geq  d\lambda/\epsilon \right) \leq 2{\epsilon},
\end{align}
as desired.
\end{proof}
\section{Extension to sub-gaussian distributions}\label{sec:subg-filter}
We now consider a variant of the filter algorithm (\autoref{alg:filter}) analyzed in Section~\ref{sec:mwu}.
The difference is that instead of fixing the step size to be $\eta = 1/2$, we set it as $\epsilon$.
That is, we will perform the multiplicative update less aggressively when there are few bad points.
In addition, we require a stronger approximation for the largest eigenvector computation. This increases the the run-time by an $O(\text{poly}(1/\epsilon))$ factor.
For technical reasons, we also ask the algorithm to stop early if the weighted covariance has been reduced to a desired value.
Formally, the algorithm is described by the pseudo-code below (\autoref{alg:filter-subg}).

\begin{algorithm}[ht!] \label{alg:filter-subg}
    \caption{Multiplicative weights   for sub-gaussian robust mean estimation}
    \KwIn{A set of   points $\{x_i\}_{i=1}^n$, an iteration count $T$, and parameter $\rho,\delta$}
\KwOut{A set of weights $w \in \mathcal{W}_{n,\epsilon}$.}
\vspace{10pt}
Let  $w^{(1)} = \frac{1}{n} \I_n$. \\
\For{$t$ from $1$ to $T$}{
    Let $\nu^{(t)} = \sum_i w^{(t)}_i x_i$, $M^{(t)} = \sum_i w^{(t)}_i
(x_i- \nu^{(t)})(x_i - \nu^{(t)})^T$.\\
Compute $v^{(t)}= \textsc{ApproxTopEigenvector}(M^{(t)}, 1-\epsilon^2, \delta/T) $.\\
    \textbf{If} $\lambda^{(t)} =v^{{(t)}\top} M^{(t)}v^{(t)} \leq 1$, \textbf{return} $w^{(t)}$.  \\
    Compute $\tau^{(t)}_i =\left \langle v^{(t)}, x_i- \nu^{(t)} \right\rangle^2$.\\
    Set $w_i^{(t+1)} \leftarrow w_i^{(t)}\left(1- \epsilon\tau^{(t)}_i/\rho\right)$ for each $i$.\\
    Project $w^{(t+1)}$ onto the set of good weights $\mathcal W_{n,\epsilon}$ (under KL divergence).
    %\Snote{Might need some clarification on what this projection means.}
}
\Return $w^{(t^*)}$, where $t^* = \argmin_t \|M^{(t)}\|$.
\end{algorithm}

First, we need a stronger spectral signature lemma.
\begin{lemma}[\cite{dong2019quantum}]\label{lem:ss=subg}
Let $S=\{x_i\}_{i=1}^n$ be an $\epsilon$-corrupted set of $n$ samples from a sub-gaussian distribution  over $\mathbb R^d$, with mean $\mu$ and identity covariance. Suppose $n \geq \widetilde \Omega (d/\epsilon^2)$. If $\|M(w)\| \leq 1+\lambda$, for some $\lambda \geq 0$, then for any $w\in \mathcal{W}_{n,2\epsilon}$, 
\begin{align*}
    \| \mu - \mu(w)\| \leq \frac{1}{1-\epsilon}\left (\sqrt{\epsilon\lambda} + C\epsilon \sqrt{\log (1/\epsilon)}\right),
\end{align*}
for some universal constant $C>0$.
\end{lemma}
Moreover, we assume that  for all $w\in \mathcal{W}_{n,2\epsilon}$  we have 
\begin{align}\label{eq:subg-cond}
    \left\| \sum_{i\in G} w_i (x_i -\mu)(x_i-\mu)^\top - I \right\| \leq \lambda=  O(\epsilon \log(1/\epsilon)).
\end{align}
This condition holds with high probability over the draws of samples~\cite{diakonikolas2019robust}. 

\begin{lemma}[analysis of sub-gaussian filter]\label{lem:subg-filter}
 Let  $\epsilon$ be a sufficiently small constant  and $\{x_i\}_{i=1}^n$ be $n$ points in $\mathbb R^d$.
Assume  the following (deterministic) conditions hold.
\begin{enumerate}[(i)]
    \item There exists $\nu\in \mathbb R^d$ and  $w \in \mathcal{W}_{n
    ,\epsilon}$ such that 
    \begin{equation} \label{eq:subg-det}
 \left\|   \sum_{i=1}^n w_{i}\left(x_{i}- \nu\right)\left(x_{i}- \nu\right)^{\top}\right\| \le 1+ O\left(\epsilon \log \left(1/\epsilon\right)\right).
    \end{equation}
 \item    If $\|M(w)\| \leq 1+\lambda$, for some $\lambda \geq 0$,  then for any $w\in \mathcal{W}_{n,\epsilon}$, 
\begin{align}\label{eq:subg-ss-det}
    \| \nu - \mu(w)\| \leq \frac{1}{1-\epsilon}\left (\sqrt{\epsilon\lambda} + C\epsilon \sqrt{\log (1/\epsilon)}\right),
\end{align}
\end{enumerate}
Then, given $\{x_i\}_{i=1}^n$, a failure rate $\delta$ and $\rho$ such that $\rho \geq \tau_i^{(t)}$ for all $i$ and $t$, \autoref{alg:filter-subg} finds   $w'\in \mathcal W_{n,\epsilon}$   such that 
\begin{equation}\label{eqn:subg-goal}
    \|M(w') \| \leq 1+ O\left(\epsilon \log \left(1/\epsilon\right)\right),
\end{equation} with probability at least $1-\delta$. 
    
The algorithm terminates in   $T = O(\rho / \epsilon)$ iterations. Further, if $T =O(\text{poly}(n,d))$, then each iteration takes $\widetilde O(nd\log \left(1/\delta)/\epsilon^2\right)$ time.
\end{lemma} 
\begin{proof}[Proof of \autoref{lem:subg-filter}]
If the algorithm gets stopped early (at Line 5), then it means that $$\|M^{(t)}\| \leq \lambda^{(t)} /\left(1-\epsilon^2\right)\leq 1/\left(1-\epsilon^2\right) \leq 1+ O(\epsilon^2),$$
since $v^{(t)}$ is a $\left(1-\epsilon^2\right)$ approximate largest eigenvector of $M^{(t)}$. Hence, in this case, we immediately achieves the goal~\eqref{eqn:subg-goal}.  

Now assume the algorithm did not stop early and so $\|M^{(t)}\| >1$ for all $t$. Then we have
\begin{equation} \label{eq:sub-g-id}
    \sum_i w_i^{(t)} \tau_i^{(t)} = \sum_i w_i^{(t)}  \left \langle v^{(t)}, x_i- \nu^{(t)}
    \right\rangle^2 = v^{(t)\top } M^{(t)}v^{(t)} \geq\left(1-\epsilon^2\right)
 \left\|M^{(t)}\right\|_2,
\end{equation}
for all $t$.
Since the step size $\epsilon <1/2$ and $\rho \geq \tau_i^{(t)}$ for all $i,t$ by assumption, we can apply the regret bound of MWU (\autoref{lem:regret}) and conclude that  for $w$ that satifies assumption~\eqref{eq:subg-det},
\begin{equation}\label{eq:subg-regret2}
     \frac{1-\epsilon^2}{T} \sum_{t=1}^T \left\|M^{(t)}\right\|_2
  \leq   \frac{1}{T} \sum_{t=1}^T \left \langle w^{(t)}  , \tau^{(t)}\right\rangle \leq (1+\epsilon) \frac{1}{T}     \sum_{t=1}^T\left \langle w , \tau^{(t)}\right \rangle +\frac{\rho\cdot  \text{KL}(w|| w^{(1)}) }{T\epsilon}.
\end{equation}
We now focus on bounding $\frac{1}{T}     \sum_{t=1}^T\left \langle w , \tau^{(t)}\right \rangle$. 
\begin{claim}\label{clm:subg}
In the setting above, we have
\begin{align*}
      \frac{1}{T}
    \sum_{t=1}^T \left\langle w , \tau^{(t)}\right\rangle \leq 1 +O\left(\epsilon\log(1/\epsilon)\right) + \frac{2\epsilon}{(1-\epsilon)^2} \frac{1}{T}\sum_{t=1}^T \left\|M^{(t)}\right\| -  \frac{2\epsilon}{(1-\epsilon)^2} 
\end{align*}
\end{claim}
\begin{proof}
Note that 
\begin{align}
    \frac{1}{T}
    \sum_{t=1}^T \left\langle w , \tau^{(t)}\right\rangle&=  \frac{1}{T} \sum_{t=1}^T  \sum_{i=1}^n w_i \left\langle x_i - \nu^{(t)}
        ,v^{(t)}\right \rangle\nonumber\\
        &= \frac{1}{T} \sum_{t=1}^T \sum_{i=1}^n   w_i \left(\left\langle x_i - \nu
    ,v^{(t)}\right \rangle^2 +  \left\langle\nu- \nu^{(t)}, v^{(t)}\right \rangle^2\right)\nonumber\\
      &\leq   1+ O\left(\epsilon\log(1/\epsilon)\right) +  \frac{1}{ T} \sum_{t=1}^T \left\langle\nu- \nu^{(t)}, v^{(t)}\right \rangle^2\label{eq:c8}\\
    &\leq 1+ O\left(\epsilon\log(1/\epsilon)\right) + \frac{1}{T} \sum_{t=1}^T  \left\| \nu - \nu^{(t)}\right\|_2^2,\label{eq:c7} 
\end{align}
where \eqref{eq:c8} follows from the assumption~\eqref{eqn:subg-goal}. Now we apply assumption \eqref{eq:subg-ss-det} to bound $\left\| \nu - \nu^{(t)}\right\|_2^2$.  Since we may assume $\|M^{(t)}\| \geq 1$ by the early stopping of Line 5, we have
\begin{align*}
    \left\|\nu - \nu^{(t)}\right\|^2 &\leq \frac{2}{(1-\epsilon)^2} \left( {\epsilon \left(\left\|M^{(t)}\right\|-1\right) + C^2\epsilon^2 {\log(1/\epsilon)}} \right) \\
    &= \frac{2\epsilon}{(1-\epsilon)^2} \left\|M^{(t)}\right\| -  \frac{2\epsilon}{(1-\epsilon)^2} + O(\epsilon^2\log(1/\epsilon)).
\end{align*}
Substituting the bound back into~\eqref{eq:c7} completes the proof. 
\end{proof}

Using \autoref{clm:subg}, the KL bound~(\autoref{lem:kl}) and~\eqref{eq:subg-regret2}, we have
\begin{align*}
       \frac{1-\epsilon^2}{T} \sum_{t=1}^T \left\|M^{(t)}\right\|_2
  \leq \frac{2(1+\epsilon)\epsilon}{(1-\epsilon)^2} \frac{1}{T}\sum_{t=1}^T \left\|M^{(t)}\right\| + 1-  \frac{2(1+\epsilon)\epsilon}{(1-\epsilon)^2}  + O(\epsilon \log(1/\epsilon))+    \frac{5\rho}{T}.
\end{align*}
For sufficiently small $\epsilon$, we   rearrange and divide through to obtain
\begin{align*}
   \frac{1}{T} \sum_{t=1}^T \left\|M^{(t)}\right\|_2\  \leq 1+ O(\epsilon\log(1/\epsilon)) + O(\epsilon)  + \frac{O(\rho)}{T}. 
\end{align*}
Setting $T = O(\rho / \epsilon)$ completes the correctness proof.  Finally, the per-iteration cost follows from the run-time of using power method to approximate the largest eigenvector. 
\end{proof}

Using the lemma we can prove our main theorem. 

\begin{theorem}[sub-gaussian robust mean estimation, \cite{diakonikolas2019robust}]\label{thm:subg-mwu}
Let $S=\{x_i\}_{i=1}^n$ be an $\epsilon$-corrupted set of $n$ samples from a sub-gaussian distribution  over $\mathbb R^d$, with mean $\mu$ and identity covariance. Suppose $n \geq \widetilde \Omega (d/\epsilon^2)$.  Given $S$, there is an algorithm that  outputs $\widehat \mu$ such that $\|\widehat \mu - \mu \| \leq  O\left(\epsilon \log \left(1/\epsilon\right)\right)$ with high constant  probability. The algorithm runs in time $\widetilde O\left(nd^2 / \epsilon^3\right)$ 
\end{theorem}
\begin{proof}
Let $\delta =0.01$.
We apply \autoref{alg:filter-subg} with a simple pruning procedure as a preprocessing. 
By standard concentration of sub-gaussian random vectors, with high constant probability,  $\|x_i - \mu\| \leq r = O(\sqrt{d \log n})$ for all $i\in G$. Hence, we apply \textsc{Prune}($S, r, \delta$), and by \autoref{lem:naive-prune} it guarantees to terminate in $O(nd)$ time and removes at most $\epsilon n$ (bad) points. 

We feed the remaining (at least) $(1-\epsilon)n$ points $R \supseteq G$ into \autoref{alg:filter-subg} with $\rho = r^2$. Notice that  $\tfrac{1}{(1-\epsilon)(1-\epsilon)} \leq \tfrac{1}{1-2\epsilon}$ for $\epsilon \le 1/2$,  so assumptions (i)-(ii) of \autoref{lem:subg-filter} are satisfied by the claim of \eqref{eq:subg-cond} and \autoref{lem:ss=subg}, respectively. 

It then follows from  \autoref{lem:subg-filter} that  
\autoref{alg:filter-subg} outputs $w' \in \mathcal{W}_{|R|,\epsilon}$
such that \[\left\| \sum_{i\in R} (x_i - \mu(w')) (x_i -\mu(w'))^\top\right\|\leq 1+O\left(\epsilon \log(1/\epsilon) \right) ,\]
where $\mu(w') = \sum_{i\in R} w_i' x_i$. 
Let  $w_i''= w_i'$ if $i\in R$ and $w_i'' = 0$ otherwise.
We obtain $w''\in \mathcal W_{n,2\epsilon}$ such that $\|M(w'')\| \leq 1+ O\left(\epsilon \log \left(1/\epsilon\right)\right)$. Applying the spectral signature (\autoref{lem:ss=subg}) proves that $\mu(w'')$ attains the desired estimation error. Moreover, the run-time simply follows from \autoref{lem:subg-filter}.
\end{proof}
\section{Sample reweighing via Matrix Multiplicative Update}
\label{sec:mmwu}

We now show that the spectral sample reweighing problem (\autoref{def:ssr}) can be solved in near linear time via a matrix multiplicative update scheme from the recent work of~\cite{dong2019quantum}, analyzed there for the robust mean estimation setting.  Our analysis  will closely resemble the arguments therein.

\begin{theorem}\label{thm:mmw}
Let  $\{x_i\}_{i=1}^n$ be $n$ points in $\mathbb R^d$.
Suppose there exists $\nu\in \mathbb R^d$ and  $w^* \in \mathbb{R}_{n}$ such that $|w^*| =1-\epsilon$, $\|w^*\|_\infty \leq 1/n$ and $
    \sum_{i=1}^n w^*_{i}\left(x_{i}- \nu\right)\left(x_{i}- \nu\right)^{\top}\preceq \lambda I$
for some   $\lambda  > 0$ and a sufficiently small $\epsilon$.   
Then, given $\{x_i\}_{i=1}^n, \lambda$, the squared diameter $\rho$ of the points and a failure rate $\delta$, there is a matrix multiplicative weights-based algorithm (Algorithm~\ref{alg:mmwu}) that, with probability at least $1-\delta$, finds   $w\in \mathcal W_{n,\epsilon}$ and   $\nu' \in \mathbb R^d$ such that $$
    \sum_{i=1}^n w_{i}\left(x_{i}- \nu'\right)\left(x_{i}- \nu'\right)^{\top}  \preceq O(\lambda)I.$$ 
    Further, the algorithm terminates in $O(\log (\rho/\lambda))$ iterations, where $\rho$ is the squared diameter of the input points $\{x_i\}_{i=1}^n$, and each iteration can be implemented in $\widetilde O(nd\log(1/\delta))$ time.
\end{theorem}
\begin{remark}
In the following, we will consider an idealized version of the algorithm and omit the detail of implementing the numerical linear algebra primitives in $\widetilde O(nd\log(1/\delta))$ time each iteration. The exact details can be found in~\cite{dong2019quantum}.
\end{remark}

The algorithm is based on the matrix multiplicative weights update. For a sequence of PSD matrices  $M_1\succeq M_2 \succeq \cdots \succeq M_{t-1}$, we will apply  the matrix multiplicative weight (MMW) update,  given by
\begin{align} \label{eq:mmw}
   \textsf{MMW}(  M_{0},M_{1},\cdots, M_{t-1} ) =\exp \left(\frac{1}{\|M_0\|_{2}} \sum_{k=1}^{t-1} M_k\right) / \operatorname{tr} \exp \left(\frac{1}{\|M_0\|_{2}} \sum_{k=1}^{t-1} M_k\right).
\end{align}  
For technical reasons, we will not maintain a set of weights that is a probability distribution throughout. Instead, we will initiate from uniform weights and monotonically downweight each point. The key invariant we will maintain is the following, which is a weighted extension of the notion of ``mostly-good weights'' from \cite{dong2019quantum}.
%Instead, recall by Lemma~\ref{lem:good} that there exists a subset $G$ of size $(1-\epsilon)n$ such that $A_G\preceq \lambda I$, where  $A_G=\frac{1}{(1-\epsilon)n}  \sum_{i\in G}  \left(x_{i}- \nu\right)\left(x_{i}- \nu\right)^{\top}$. Thus,  our new notion of a good set of weights  is that starting from the uniform distribution over $n$ points,   more weights are removed from $[n]\setminus G$ than from $G$. Let $w_G, w_B$ denote the restriction of $G$ to the indices of vector $w$ and $B = [n] \setminus G$.
\begin{definition}[mostly-good weight]\label{def:mgod}
 Suppose that $w^*\in [0,1]^n$ satisfies  $\|w^*\|_\infty \leq 1/n$.
 The set of mostly-good weight vectors (with respect to $w^*$) is  \[ \mathcal{C}(w^*)=\left\{ w \in \mathbb R^n: 0 \leq w_i \leq \frac 1 n  \quad \text{and} \quad  \sum_{i=1}^n w_i^* \left(\frac{1}{n} - w_i\right) \leq \sum_{i=1}^n \left(\frac{1}{n} - w_i^*\right) \left(\frac{1}{n} - w_i\right)\right\}\]
\end{definition}
\begin{lemma}\label{lem:big}
Suppose that $w^*\in [0,1]^n$ satisfies  $\|w^*\|_\infty \leq 1/n$ and $|w^*| =1-\epsilon$. Then for any mostly-good weight $w \in \mathcal{C}(w^*)$ (with respect to $w^*$), we have that $|w|\geq 1-2\epsilon$.
\end{lemma}
\begin{proof}
By rearranging the condition of mostly-good weight, we get that 
\begin{align*}
    \frac{1}{n}\sum_{i=1}^n \frac{1}{n}-w_i\geq 2\sum_{i=1}^n w_i^* \left(\frac{1}{n} - w_i\right).
\end{align*}
Since $\sum_{i=1}^n w_i^* = 1-\epsilon$, it follows that 
\begin{align*}
    \frac{1}{n} - \frac{1}{n}\sum_{i=1}^n w_i \geq \frac{2-2\epsilon}{n} - 2\sum_{i=1}^n w_i^* w_i.
\end{align*}
By assumption, $w_i^*\leq \frac{1}{n}$, so
\begin{align*}
    \frac{1}{n} - \frac{1}{n}\sum_{i=1}^n w_i \geq \frac{2-2\epsilon}{n} - \frac{2}{n}\sum_{i=1}^n  w_i.
\end{align*}
Multiplying $n$ on both sides and rearranging we get $\sum_{i=1}^n w_i \geq 1-2\epsilon$.
\end{proof}
A crucial subroutine we use is a deterministic down-weighting scheme,  from~\cite{dong2019quantum}, that maintains the mostly-good property of the input weights. 
\begin{lemma}[1D Filter~\cite{dong2019quantum}]\label{lem:1dfilter}
 Let $\eta \in (0, 1/2)$, let $b \geq 2 \eta$, and let $w_1, \ldots, w_m$ and $\tau_1, \ldots, \tau_m$ be non-negative numbers so that $\sum_{i = 1}^m w_i \leq 1$.
Let $\tau_{\max} = \max_{i \in [m]} \tau_i$.
Suppose there exists $w^*$ such that $\|w\|_\infty \leq \frac{1}{n}$
\[
\sum_{i=1}^n w_i^* w_i \tau_i \leq \eta \sigma\; \mbox{, where } \; \sigma = \frac{1}{n} \sum_{i = 1}^n w_i \tau_i \; .
\]
Then $\textsc{1DFilter} (w, \tau, b)$ runs in time $O((1 + \log \frac{\tau_{\max}}{b \sigma}) m)$ and outputs $0 \leq w' \leq w$ so that
\begin{itemize}
\item  $\sum w_i^* (w_i - w_i')  \le   \sum (1/n - w_i^*) (w_i - w_i')  
$, and
\item 
$\frac{1}{n}\sum_{i = 1}^n  w'_i \tau_i \leq b \sigma.$
\end{itemize}
 \end{lemma}
The algorithm is formally described in Algorithm~\ref{alg:mmwu}. 
Throughout   let $M^{(s)} = M(w^{(s)})$ and $M^{(s)}_t= M(w^{(s)}_t)$, where  $M(w) = \sum_{i=1}^n w_i (x_i - \mu(w))(x_i-\mu(w))^\top$. The procedure runs by epochs, where each epoch $s$   reduces the largest eigenvalue of $M^{(s)}$ by a constant factor. We will show  that the inner loop achieves the reduction within $O(\log d)$ iterations while maintaining the invariant that the weights are mostly-good (\autoref{def:mgod}). 
\begin{algorithm}[ht!] \label{alg:mmwu}
    \caption{Matrix multiplicative update for spectral sample reweighing (\autoref{def:ssr})}
    \KwIn{A set of   points $x_1, \ldots, x_n$, $\lambda,\rho$ and a failure rate $\delta$}
\KwOut{A point $\nu'\in \mathbb R^d$  and weights $w'\in \mathcal{W}_{n,\epsilon}$ that satisfy~\eqref{eq:goal} up to a constant factor.}

\vspace{10pt}
Let  $w^{(0)} = \frac{1}{n} (1,1,\cdots, 1)$. \\
\For{$s$ from $0$ to $O(\log \rho)$}{
    Compute $\lambda^{(s)}  =\|M^{(s)}\|$. \\
    \If{$\lambda^{(s)} \leq 300\lambda$}{
        \Return $w^{(s)}/\|w^{(s)}\|_1, \mu(w^{(s)})$.
    }
    \For{$t$ from $0$ to $O(\log d)$}{
        Compute $\lambda^{(s)}_t  =\|M_t^{(s)}\|$ and terminate epoch \textbf{if} $\lambda^{(s)}_t \leq \frac 2 3 \lambda^{(s)}_0$.\\
        Compute  $U_t^{(s)} =\textsf{MMW}(M_1^{(s)},M_2^{(s)},\cdots, M_{t-1}^{(s)}) $.\\
        Compute
        \begin{equation}
                    \tau_{t,i}^{(s)} = \left(x_{i}-\mu\left(w_{t}^{(s)}\right)\right)^{\top} U_{t}^{(s)}\left(x_{i}-\mu\left(w_{t}^{(s)}\right)\right)
        \end{equation}\\
        Let $w_{t+1}^{(s)}=w_{t}^{(s)}$ \textbf{if} $\sum_i w_{t,i}^{(s)} \tau_{t,i}^{(s)}  \leq \frac 1 4 \lambda_1^{(s)}$; \textbf{otherwise} $w_{t+1}^{(s)} = \textsc{1DFilter}(w_t^{(s)}, \tau_{t}^{(s)})$.
    }
    Let $w^{(s+1)}=w_{t}^{(s)}$.
}
\end{algorithm}
 Similar to our MWU analysis, our argument relies on a spectral signature lemma. 
\begin{lemma}[spectral signature for mostly-good weights] \label{lem:ssc}
Let  $\{x_i\}_{i=1}^n$ be $n$ points in $\mathbb R^d$.
Suppose there exists $\nu\in \mathbb R^d$ and  $w^* \in \mathbb{R}_{n}$ such that $|w^*| =1-\epsilon$ for sufficiently small $\epsilon$, $\|w^*\|_\infty \leq 1/n$ and $$
    \sum_{i=1}^n w^*_{i}\left(x_{i}- \nu\right)\left(x_{i}- \nu\right)^{\top}\preceq \lambda I$$
for some   $\lambda  > 0$. Then  for any $w\in \mathcal{C}(w^*)$, 
    \begin{equation}\label{eq:ss3}
        \|\nu- \nu(w)\| \le  \frac{1}{1-\sqrt{2\epsilon}} \left(3\sqrt{\lambda}+    2\sqrt{\epsilon\|M(w)\|}\right),
    \end{equation}
    where $\nu(w) = \sum_{i}w_i x_i$ and $M(w) = \sum_i w_i (x_i - \nu(w))(x_i - \nu(w'))^\top$.
\end{lemma}
\begin{proof}
This directly follows from \autoref{lem:ssw} and scaling.
\end{proof}
Using this, we establish a key invariant of the inner loop of the algorithm.
 \begin{lemma}\label{lem:sscc}
Let $w \in\mathcal{C}(w^*)$ be such that $\beta= \|M(w)\|_2 \geq 300\lambda$ and $U$ be a density matrix. Let $\tau_i= \left(x_{i}-\mu\left(w\right)\right)^{\top} U\left(x_{i}-\mu\left(w\right)\right)$. If $w' = \textsc{1DFilter}(w, \tau, 1/4)$, then we have $w'\in \mathcal{C}(w^*)$ and $\left\langle M\left(w^{\prime}\right), U\right\rangle \leq \frac 1 4\langle M(w), {U}\rangle$.
\end{lemma}
\begin{proof} Let $\mu( w^* ) = \sum_i  w_i^* x_i$.
Then for any unit vector $u$, we have that by Jensen's inequality
\begin{align*}
     \langle\mu( w^*) - \nu,u\rangle^2 \leq   \left\langle   \sum_{i=1}^n w_i^* x_i-\nu, u\right\rangle^2 \leq   \sum_{i=1}^n w_i^* \langle x_i-\nu , u \rangle^2 \leq \lambda.
    \end{align*}
%    \Snote{shouldn't $w_i^*$ still appear in the last inequality above? why switch to $1/n$? Seems false as written.} \fred{$1/n$ should be $w_i^*$. }
    Thus, $\|\mu({w^*})-\nu\|_{2}^{2} \leq \lambda$.
Expanding the definition of $\tau_i$, we get
\begin{align}
\sum_{i=1}^n w^*_i w_{i} \tau_{i} &=\left\langle\sum_{i =1}^n w_i^*w_{i}\left(x_{i}-\mu(w)\right)\left(x_{i}-\mu(w)\right)^{\top}, U\right\rangle \nonumber\\
&\le \frac{1}{n}\left\langle\sum_{i=1}^{n} w^*_{i}\left(x_{i}-\mu(w^*)\right)\left(x_{i}-\mu(w^*)\right)^{\top}, U\right\rangle\nonumber\\&\quad \quad\quad +\frac{1}{n}\|w^*\|_1 \cdot(\mu(w^*)-\mu(w))^{\top} U(\mu(w^*)-\mu(w)) \nonumber\\
&\leq \frac{1}{n} \langle M(w^*), U\rangle+\frac{1}{n}(1-\epsilon)\|\mu(w^*)-\mu(w)\|_{2}^{2} \nonumber \\
&\leq \frac{1}{n}\lambda + \frac2 n\|\mu(w^*)-\nu\|_{2}^{2}+\frac{6}{n}\|\mu(w)-\nu\|_{2}^{2}\label{eq:hl2}\\\
&\leq \frac{1}{n}\lambda + \left(\frac{4}{n}\lambda + \frac{2}{n}\epsilon \|M(w)\|)\right) \label{eq:hl1}\\
&\leq \frac 1 {30n} \|M(w)\| = \frac{1}{30n}\sum_{i=1}^n  w_i \tau_i,\label{eq:hl3}
\end{align}
%\Snote{why switch to capital $X_i$'s above? $x_i$ right?}\fred{yep, fixed.}
where \eqref{eq:hl2} follows  from the spectral centrality condition and triangle inequality,
%\Snote{isn't it just triangle inequality, where $\nu$ is as in the centrality condition?} \fred{there is this very first term}
\eqref{eq:hl1} follows from~\autoref{lem:ssc},
%\Snote{this step might have to be adjusted if the spectral signature lemma in this section is adjusted}\fred{yes, and for sufficiently small $\epsilon$ we should be fine.}
and~\eqref{eq:hl3} uses our assumption that $\|M(w)\| \geq 300\lambda$, the definition of $\tau_i$ and $\epsilon$ is sufficiently small. This allows us to  apply the guarantee of the 1D filter procedure (\autoref{lem:1dfilter}) and get that

\begin{align*}
    \left\langle M\left(w^{\prime}\right), U\right\rangle =\left\langle \sum_{i=1}^{n} w_{i}^{\prime}\left(X_{i}-\mu(w')\right)\left(X_{i}-\mu(w')\right), U \right\rangle=\sum_{i=1}^{n} w_{i}^{\prime} \tau_{i} \leq \frac{1}{4}\sum_{i=1}^{n} w_{i} \tau_{i} = \frac{1}{4} \langle M(w), U\rangle.
\end{align*}
Furthermore, $w' \in \mathcal{C}(w^*)$. This completes the proof. 
\end{proof}

We are now ready to prove the main theorem of this section. 
\begin{proof}[Proof of~\autoref{thm:mmw}]
Consider a fixed epoch and drop the super script for simplicity of notation. It is not hard to observe that $M(w_{t+1}) \preceq M(w_{t})$ (see Lemma 3.4~\cite{dong2019quantum}). Let $\alpha = 1/\|M(w_0)\|$. A regret bound for matrix multiplicative weights~\cite{allen2015spectral} implies that 
\begin{align*}
\left\|\sum_{t=0}^{T-1} M\left(w_{t+1}\right)\right\|_{2} & \leq \sum_{t=0}^{T-1}\left\langle M\left(w_{t+1}\right), U_{t}\right\rangle+\alpha \sum_{t=0}^{T-1}\left\langle U_{t}, M\left(w_{t+1}\right)\right\rangle\left\|M\left(w_{t+1}\right)\right\|_{2}+\frac{\log d}{\alpha} \\
& \leq 2 \sum_{t=0}^{T-1}\left\langle M\left(w_{t+1}\right), U_{t}\right\rangle+\left\|M\left(w_{0}\right)\right\|_{2} \cdot \log d
\end{align*}
Now   by definition of Line 10,  we have $\left\langle M\left(w_{t+1}\right), U_{t}\right\rangle \leq \frac{1}{4}\left\|M\left(w_{0}\right)\right\|_{2}$. Hence,
\begin{align*}
  T\left\|M\left(w_{T}\right)\right\|_{2}\leq   \left\|\sum_{t=0}^{T-1} M\left(w_{t}\right)\right\|_{2} \leq T \cdot \frac{1}{2}\left\|M\left(w_{0}\right)\right\|_{2}+\left\|M\left(w_{0}\right)\right\|_{2} \cdot \log d.
\end{align*}
Setting $T \gg \log d$   shows that the inner loop terminates in $O(\log d)$ iterations and reduces the largest eigenvalue of the covariance by, say, $4/5$. 

Finally, to bound the number of epochs, we simply note that $\|M^{(0)}\| \leq \rho$. Therefore, $O(\log (\rho/\lambda))$ epochs suffice drive the largest eigenvalue of $\|M^{(s)}\|$ down to $O(\lambda)$, since it is reduced geometrically each epoch.
\end{proof}
% \paragraph{An extension to sub-gaussian setting} We remark that one can consider a strengthened, \textit{two-sided} version of the spectral sample reweighing problem (\autoref{def:ssr}), where the centrality condition~\eqref{eq:p} is replaced by
% \begin{equation*}
%     \left\| \sum_{i=1}^n w_{i}\left(x_{i}- \nu\right)\left(x_{i}- \nu\right)^{\top} - I\right\| \leq  \lambda.  
% \end{equation*}
% By the same line of argument, a similar matrix multiplicative weight algorithm from~\cite{dong2019quantum} can be used to find $w',\nu'$ that approximately satisfies the condition.  This enables robust mean estimation for sub-gaussian distributions. However, since this formulation has no down-stream implications, we do not emphasize or expand this point.
\section{Sample reweighing via Online Gradient Descent}
\label{sec:gd}

\subsection{Regret analysis of gradient descent}
We now consider a gradient updated-based algorithm for solving the spectral sample reweighing problem (\autoref{def:ssr}). The analysis will be through the classic regret guarantee of online gradient descent for convex optimization~\cite{zinkevich2003online}. Though the resulting run-time is higher than the MWU scheme we analyzed in Section~\ref{sec:mwu}, it nonetheless betters the recent work of~\cite{cheng2020high}, where essentially the same gradient descent-based algorithm is studied.

We will leverage the following regret guarantee of online gradient descent; the definition of the algorithm in the general setting can be found in standard text~\cite{hazan2016introduction}. 
\begin{lemma}[Theorem 3.1~\cite{hazan2016introduction}, originally due to~\cite{zinkevich2003online}]\label{lem:gd-regret}
Let $f_t: \mathcal{K} \rightarrow \mathbb R$ be the convex cost function revealed at iteration $t$, where $\mathcal{K}$ is a convex feasible set.  Suppose $f_t$ is $L$-Lipschitz (in $\ell_2$ norm) and $\|x_0 -x^*\|_2 \leq R$ for some $x^* \in \argmin_{x \in\mathcal K} \sum_t f_t(x)$. The online gradient descent algorithm with step sizes $ \eta_{t}=\frac{R}{L \sqrt{t}}$ achieves 
\begin{equation}\label{eq:gd-regret}
    \sum_{t=1}^T f_t(x_t) - \min_{x \in\mathcal K} \sum_{t=1}^T f_t(x) \leq \frac{3}{2} LR\sqrt{T}.
\end{equation}
\end{lemma}

Our algorithm implicitly defines the cost functions $f_t(w) = \left\langle w, \tau^{(t)}\right\rangle$, where the feasible set is $\mathcal W_{n,\epsilon}$, and   implements the online gradient descent algorithm for the linear objective.  Note that $\nabla f_t(w) = \tau^{(t)}$, and the main difference of this algorithm from the MWU scheme~(\autoref{alg:filter}) is that we use an additive/gradient-descent update, in lieu of the multiplicative update.

\begin{algorithm}[ht!] \label{alg:gd}
    \caption{Gradient descent for spectral sample reweighing (\autoref{def:ssr})}
        \KwIn{A set of   points $\{x_i\}_{i=1}^n$, an iteration count $T$, and step sizes $\eta_t$}
\KwOut{A point $\nu\in \mathbb R^d$  and weights $w\in \mathcal{W}_{n,\epsilon}$.}
\vspace{8pt}
Let  $w^{(1)} = \frac{1}{n} (1,1,\cdots, 1)$. \\
\For{$t$ from $1$ to $T$}{
    Let $\nu^{(t)} = \sum_i w^{(t)}_i x_i$, $M^{(t)} = \sum_i w^{(t)}_i
(x_i- \nu^{(t)})(x_i - \nu^{(t)})^T$.\\
Let $v^{(t)} $ be the top eigenvector of $M^{(t)}$ (with $\|v^{(t)}\|=1$).\\
    Compute $\tau^{(t)}_i =\left \langle v^{(t)}, x_i- \nu^{(t)} \right\rangle^2$.\\
    Set $w_i  \leftarrow w_i - \eta_t \tau^{(t)}$.\\
    Project $w^{(t+1)}$ onto the set of good weights $\mathcal W_{n,\epsilon}$ (under $\ell_2$ distance).
    %\Snote{Might need some clarification on what this projection means.}
}
\Return $\nu^{(t^*)}, w^{(t^*)}$, where $t^* = \argmin_t \|M^{(t)}\|$.
\end{algorithm}
\begin{lemma}\label{lem:lip}
Let $\rho$ be the squared diameter of the inputs points $\{x_i\}_{i=1}^n$. The cost function $f_t(\cdot)$ is $\sqrt{n}\rho$-Lipschitz (in $\ell_2$ norm), for all $t$. 
\end{lemma}
\begin{proof}
Since $f_t$ is differentiable, we only need the bound $\|\nabla f_t\|$.   We have that for all $t$ and $i$,
\begin{align*}
    \tau^{(t)}_i =\left \langle v^{(t)}, x_i- \nu^{(t)} \right\rangle^2 \leq \|   x_i- \nu^{(t)}\|^2_2 \leq \rho.
\end{align*}
Therefore, $\|\nabla f_t\| = \|\tau^{(t)}\| \leq \sqrt{n}\rho$.
\end{proof}

\begin{theorem}\label{thm:gdbound}
Given $\{x_i\}_{i=1}^n$ and $\eta_t = R/L\sqrt{t}$ with $L=\sqrt{n}\rho, R=\sqrt{2}$, the online gradient descent algorithm (based on \autoref{alg:gd}) yields a constant-factor approximation for the spectral sample reweighing problem (\autoref{def:ssr}) in $O(n d^2/\epsilon^2)$ iterations and $O(n^2d^3/\epsilon^2)$ total run-time.
\end{theorem}
\begin{proof}

We first apply the \textsc{Prune} procedure of~\autoref{lem:naive-prune} to bound the diameter. By~\autoref{lem:diameter} and the guarantee of \textsc{Prune}, we can have $\rho= 16 d\lambda/\epsilon$. Then we apply \autoref{alg:gd}.

We will use~\autoref{lem:gd-regret} to analyze \autoref{alg:gd}. First, by~\autoref{lem:lip}, we have $L = \sqrt{n}\rho$, and further,
since the $\ell_2$ diameter of the probability simplex can be (trivially) bounded by $\sqrt{2}$, $R = \sqrt{2}$.   
Moreover, observe for any $t$,
\begin{align*}
    f_t\left(w^{(t)}\right) = \left\langle w^{(t)}, \tau^{(t)}\right\rangle =\sum_i w_i^{(t)}  \left \langle v^{(t)}, x_i- \nu^{(t)}
    \right\rangle^2 = v^{(t)T } M^{(t)}v^{(t)} = \left\|M^{(t)}\right\|_2.
\end{align*}
Let $w\in \mathcal W_{n,\epsilon}$ be a weight that satisfies the spectral centrality condition. Then, from the regret guarantee~\eqref{eq:gd-regret},  
\begin{align}\label{eq:gd2}
    \frac1T \sum_{t=1}^T \left\|M^{(t)}\right\|_2 \leq \frac1T \sum_{t=1}^T  \left\langle w, \tau^{(t)}\right\rangle +   \frac{3LR}{2\sqrt{T}} 
\end{align}
We bound the two terms on the right side individually.
\begin{enumerate}[(i)]
    \item A bound on the first term follows exactly from the calculations we did in the analysis of MWU algorithm (\autoref{alg:filter}). In particular, from~\eqref{eq:right-2}  we have 
    \begin{align*}
        \frac{1}{T}
    \sum_{t=1}^T \langle w , \tau^{(t)}\rangle&\leq  15 \lambda+ \frac{1}{3T}\sum_{t=1}^T \left\|M^{(t)}\right\|_2.
    \end{align*}
    \item  Observe that it suffices to set $T = 3L^2R^2 / \lambda^2$ to bound the second term by $\lambda$.
\end{enumerate}
Substituting the two bounds back into~\eqref{eq:gd2}, 
\begin{align} 
    \frac1T \sum_{t=1}^T \left\|M^{(t)}\right\|_2 \leq 16 \lambda+ \frac{1}{3T}\sum_{t=1}^T \left\|M^{(t)}\right\|_2. 
\end{align}
Rearranging and dividing through immediately yields the desired guarantee.

Given that $L = \sqrt{n}\rho, R=\sqrt{2}$, we have that the iteration count  $T = 6n\rho^2/\lambda^2$. Since $\rho = 16d\lambda /\epsilon$, $T = O(nd^2/  \epsilon^2)$. For the run-time,   note that instead of computing the exact largest eigenvector, we can use power method to find  an $7/8$-approximate one. Observe that this suffices for our analysis of the method above. Finally, the Euclidean projection onto $\mathcal{W}_{n,\epsilon}$ can be computed in $O(n\log n)$ time~\cite{wang2015projection}. This yields the desired run-time.
\end{proof}

\subsection{Extension to sub-gaussian setting}\label{sec:subg-gd}
\autoref{thm:gdbound} implies that a gradient descent-based algorithm (\autoref{alg:gd}) can be used for robust mean estimation under bounded covariance. We now extend the result to the sub-gaussian setting, showing that the same iteration and run-time complexity holds. The  optimal estimation error we will aim for is $O(\epsilon\sqrt{ \log(1/\epsilon)})$.  We assume the spectral signature~\autoref{lem:ss=subg} and the deterministic condition~\eqref{eq:subg-cond}.

%As we need the stronger end goal of $\|M(w)\| \leq 1+ O\left(\epsilon \log \left(1/\epsilon\right)\right)$ by~\autoref{lem:ss=subg}, a constant-factor approximation of spectral sample reweighing is not sufficient to capture this setting.
%We show that a slight variant of \autoref{alg:gd} achieves our goal.  

\begin{algorithm}[ht!] \label{alg:gd-subg}
    \caption{Gradient descent   for sub-gaussian robust mean estimation}
    \KwIn{A set of   points $\{x_i\}_{i=1}^n$, step sizes $\eta_t$, an iteration count $T$, and parameter $\rho$}
\KwOut{A set of weights $w \in \mathcal{W}_{n,\epsilon}$.}
\vspace{10pt}
Let  $w^{(1)} = \frac{1}{n} \I_n$. \\
\For{$t$ from $1$ to $T$}{
    Let $\nu^{(t)} = \sum_i w^{(t)}_i x_i$, $M^{(t)} = \sum_i w^{(t)}_i
(x_i- \nu^{(t)})(x_i - \nu^{(t)})^T$.\\
Compute $v^{(t)}= \textsc{ApproxTopEigenvector}(M^{(t)}, 1-\epsilon^2, \delta/T) $.\\
    \textbf{If} $\lambda^{(t)} =v^{{(t)}\top} M^{(t)}v^{(t)} \leq 1$, \textbf{return} $w^{(t)}$.  \\
    Compute $\tau^{(t)}_i =\left \langle v^{(t)}, x_i- \nu^{(t)} \right\rangle^2$.\\
    Set $w_i  \leftarrow w_i - \eta_t \tau^{(t)}$.\\
    Project $w^{(t+1)}$ onto the set of good weights $\mathcal W_{n,\epsilon}$ (under $\ell_2$ distance).
 }
\Return $w^{(t^*)}$, where $t^* = \argmin_t \|M^{(t)}\|$.
\end{algorithm}

In particular, we will analyze \autoref{alg:gd-subg} and prove the following set of guarantees. 
\begin{lemma}\label{lem:subg-gd}
 Let  $\epsilon$ be a sufficiently small constant  and $\{x_i\}_{i=1}^n$ be $n$ points in $\mathbb R^d$.
Assume  the following (deterministic) conditions hold.
\begin{enumerate}[(i)]
    \item There exists $\nu\in \mathbb R^d$ and  $w \in \mathcal{W}_{n
    ,\epsilon}$ such that 
    \begin{equation} \label{eq:subg-det2}
 \left\|   \sum_{i=1}^n w_{i}\left(x_{i}- \nu\right)\left(x_{i}- \nu\right)^{\top}\right\| \le 1+ O\left(\epsilon \log \left(1/\epsilon\right)\right).
    \end{equation}
 \item    If $\|M(w)\| \leq 1+\lambda$, for some $\lambda \geq 0$,  then for any $w\in \mathcal{W}_{n,\epsilon}$, 
\begin{align}\label{eq:subg-ss-det2}
    \| \nu - \mu(w)\| \leq \frac{1}{1-\epsilon}\left (\sqrt{\epsilon\lambda} + C\epsilon \sqrt{\log (1/\epsilon)}\right),
\end{align}
\end{enumerate}
Then, given $\{x_i\}_{i=1}^n$, a failure rate $\delta$ and $\rho$ such that $\rho \geq \tau_i^{(t)}$ for all $i$ and $t$, \autoref{alg:gd-subg} finds   $w'\in \mathcal W_{n,\epsilon}$   such that 
\begin{equation}\label{eqn:subg-goal2}
    \|M(w') \| \leq 1+ O\left(\epsilon \log \left(1/\epsilon\right)\right),
\end{equation} with probability at least $1-\delta$. 
    
The algorithm terminates in   $T = O(n\rho^2 / \epsilon^2)$ iterations. Further, if $T =O(\text{poly}(n,d))$, then each iteration takes $\widetilde O(nd\log \left(1/\delta)/\epsilon^2\right)$ time.
\end{lemma}

\begin{proof}
If the algorithm gets early stopped, then $\|M^{(t)}\| \leq 1+O(\epsilon^2)$, so assumption \eqref{eq:subg-det2} guarantees that $\mu(w^{(t)})$ achieves the desired bound~\eqref{eqn:subg-goal2}.
We now assume that $\|M^{(t)}\| > 1$ for any $t$. 

By the regret bound (\autoref{lem:gd-regret}) and the inequality    $\left\langle w^{(t)}, \tau^{(t)}\right\rangle \geq \left(1-\epsilon^2\right) \left\|M^{(t)}\right\|_2$,  for a $w$ that satisfies assumption \eqref{eq:subg-det2}
\begin{equation}\label{eq:gd-regret2}
        \frac{1-\epsilon^2}T \sum_{t=1}^T \left\|M^{(t)}\right\|_2 \leq \frac1T \sum_{t=1}^T  \left\langle w, \tau^{(t)}\right\rangle +   \frac{3LR}{2\sqrt{T}},
\end{equation}
 where $L = \sqrt{n} \rho$ and $R =\sqrt{2}$. For the first term, note that we may apply \autoref{clm:subg} and obtain 
\begin{align*}
      \frac{1}{T}
    \sum_{t=1}^T \left\langle w , \tau^{(t)}\right\rangle \leq 1+ O\left(\epsilon\log(1/\epsilon)\right) + \frac{2\epsilon}{(1-\epsilon)^2} \frac{1}{T}\sum_{t=1}^T \left\|M^{(t)}\right\| -  \frac{2\epsilon}{(1-\epsilon)^2} 
\end{align*}
By setting $T = 3 L^2 R^2 /\epsilon^2 = O(n\rho^2 /\epsilon^2)$, we can bound the second term by $ O(\epsilon)$

Substituting the bounds back into~\eqref{eq:gd-regret2}, we obtain
\begin{align*}
      \frac{1-\epsilon^2}{T}
        \sum_{t=1}^T \left\|M^{(t)}\right\|_2  \leq    1 - \frac{2\epsilon}{(1-\epsilon)^2}  + O(\epsilon \log(1/\epsilon)) +\frac{1}{ T} \sum_{t=1}^T  \frac{2\epsilon}{(1-\epsilon)^2}\left\|M^{(t)}\right\|
\end{align*}
For sufficiently small $\epsilon$, we can move the last term to the left side and divide through. This immediately yields that
\begin{align*}
      \frac{1}{T}
        \sum_{t=1}^T \left\|M^{(t)}\right\|_2  \leq    1  + O(\epsilon \log(1/\epsilon)).
\end{align*}
The run-time  follows from the cost of computing  $(1-\epsilon^2)$-approximate largest eigenvector via power iteration.
\end{proof}

 Using the same argument for \autoref{thm:subg-mwu},  \autoref{lem:subg-gd} implies the following theorem.

\begin{theorem}
Let $S= \{x_i\}_{i=1}^n$ be an $\epsilon$-corrupted set of $n$ samples from a sub-gaussian distribution over $\mathbb R^d$,  with mean $\mu$ and identity covariance. Suppose $n \geq \widetilde \Omega (d/\epsilon^2)$.  Then given $S$, there is an algorithm (based on \autoref{alg:gd-subg}) that   finds $\widehat \mu$ such that  with high constant probability $\|\widehat \mu - \mu\| \leq O\left(\epsilon \sqrt{\log(1/\epsilon)}\right)$.

The algorithm runs in $\widetilde O(nd^2/\epsilon^2)$ iterations and $\widetilde O(n^2d^3/\epsilon^2)$ total time.
\end{theorem}

%%%%%%%%%%%%%%%%%%%%%%%%%%%%%%%%%%%%%%%%%%%%%%%%%%%%%%%%%%%%%%%%%%%%%
%%%%%%%%%%%%%%%%%%%%%%%%%%%%%%%%%%%%%%%%%%%%%%%%%%%%%%%%%%%%%%%%%%%%%
%%%%%%%%%%%%%%%%%%%%%%%%%%%%%%%%%%%%%%%%%%%%%%%%%%%%%%%%%%%%%%%%%%%%%
%%%%%%%%%%%%%%%%%%%%%%%%%%%%%%%%%%%%%%%%%%%%%%%%%%%%%%%%%%%%%%%%%%%%%

\subsection{Equivalence with~\texorpdfstring{\cite{cheng2020high}}{[CDGS20]}} 
The recent work of Cheng, Diakonikolas, Ge and Soltanolkotabi \cite{cheng2020high} studies a gradient-descent-based algorithm for solving the following non-convex formulation of robust mean estimation. 
\begin{align*}
    \min\,\, \| \Sigma_w \| \quad \text{such that } w\in \mathcal{W}_{n,\epsilon}.
\end{align*}
where $\Sigma_w = \sum_{i=1}^n w_i (x_i - \mu(w) ) (x- \mu(w))^\top$.  This is equivalent to
\begin{align*}
        \min_w\,\, \max_{u\in \mathbb S^{d-1}}\,\, F(w,u) = u^\top \Sigma_w u \quad \text{such that } w\in \mathcal{W}_{n,\epsilon}.
\end{align*}
The sub-gradient of $F(w,u)$ with respect to $w$ (for a fixed $u$) is given by
\begin{align}\label{eq:grad}
    \nabla_{w} F(w, u)=X  u \odot X  u-2\left(w^{\top} X u\right) X  u,
\end{align}
where $X \in \mathbb R^{n\times d}$ is the data matrix whose   the $i$th row is $x_i$.

Based on the observation, they consider and analyze an algorithm that computes a  (approximately) maximizing $u$ and performs a projected gradient descent on $w$ each iteration. 

Since \autoref{alg:gd} can be directly applied to the same robust setting (\autoref{thm:robust-gd}), it is natural to consider the relationships between the two algorithms. 
Indeed, one can argue that they are essentially the same.
First, we unpack our gradient update (\textit{i.e.}, the spectral scores) of   iteration $t$. Note that
\begin{align*}
   \nabla_i f_t(w^{(t)}) = \tau_i^{(t)} &=\left \langle v^{(t)}, x_i- \nu^{(t)} \right\rangle^2\\
   &= \left\langle v^{(t)}, x_i \right\rangle^2 + \left\langle v^{(t)}, \nu^{(t)} \right\rangle^2 - 2\left\langle v^{(t)},x_i \right\rangle \left\langle v^{(t)}, \nu^{(t)}\right\rangle \\
   &= \left(Xv^{(t)} \odot  Xv^{(t)}\right)_i + \left(w^{(t)\top} X v^{(t)}\right)^2 - 2 \left(w^{(t) \top }Xv^{(t)}\right) \left(Xv^{(t)}\right)_i
\end{align*}
since $\nu^{(t)} = \sum_i w_i^{(t)} x_i = X^T w^{(t)}$, where $\odot$ denotes entrywise product of vectors. Let $C_t = w^{(t)\top} X v^{(t)}$. Therefore, we can rewrite the gradient as
\begin{align*}
    \nabla f_t(w^{(t)})  =C_t^2\cdot \I_n + Xv^{(t)} \odot  Xv^{(t)}  - 2 C_t\cdot  Xv^{(t)}
\end{align*}
Note that the gradient~\eqref{eq:grad} used in~\cite{cheng2020high} is exactly the same as above, except without the  term of all-one vector $C_t^2\cdot \I_n $.  In the gradient update step, the additional term reduces the weight of every point uniformly by the same quantity $C_t^2$.
However, observe that by Pythagorean theorem, the (Euclidean) projection onto $\mathcal{W}_{n,\epsilon}$ can be decomposed into two (sequential)  steps: (1) first an orthogonal projection onto the affine subspace containing $\mathcal{W}_{n,\epsilon}$, and then (2) a projection onto  $\mathcal{W}_{n,\epsilon}$ itself.
Note that reducing  each coordinate by the same quantity or not  results in the same vector by the first step.   Therefore, the two algorithms yield the same sequence of iterates $(w^{(t)})_t$.
\section{Optimal Breakdown Point Analysis}\label{sec:breakdown}
We now consider a slight variant of the filter algorithm and show that it achieves the optimal breakdown point of $\epsilon = 1/2$, for the  robust mean estimation problem. Recall that both the classic filter algorithm and our \autoref{alg:filter-subg} work with the spectral scores defined as $\tau_i = \left(\left \langle v^{(t)}, x_i\right \rangle - \left\langle v^{(t)},\nu^{(t)} \right\rangle\right)^2$, where the second term is the (weighted) average of the first. Instead, the following variant  replaces that by the median.

Throughout we let  $\nu^{(t)} = \sum_i w^{(t)}_i x_i$, $M^{(t)} = \sum_i w^{(t)}_i
(x_i- \nu^{(t)})(x_i - \nu^{(t)})^T$.
\begin{algorithm}[ht!] \label{alg:breakdown}
    \caption{Optimal filter  for spectral sample reweighing~(\autoref{def:ssr})}
    \KwIn{A set of   points $\{x_i\}_{i=1}^n$, an iteration count $T$, and parameter $\delta$}
\KwOut{A point $\nu\in \mathbb R^d$  and weights $w\in \mathcal{C}_{n,\epsilon}$.}
\vspace{10pt}
Let  $w^{(1)} = \frac{1}{n} \I_n$. \\
\While{$\|M^{(t)}\| \geq \frac{16}{7} \lambda \left( 1+ \frac{1}{1/2-\epsilon} \right)$}{
    Compute $v^{(t)} = \textsc{ApproxTopEigenvector}(M^{(t)}, 7/8, \delta/T)$.\\
    Compute $\alpha_i^{(t)} = \left \langle v^{(t)}, x_i \right\rangle$ for each $i$ and let $m^{(t)} = \textsf{median} \left( \{ \alpha^{(t)}_i\}_{i=1}^n \right)$\\
    Compute $\tau^{(t)}_i = \left(\alpha_i^{(t)} - m^{(t)}\right)^2$ for each $i$ and $\tau_{\max} = \max_{i: w_i >0} \tau_{i}^{(t)}$.\\
    Set $w_i^{(t+1)} \leftarrow w_i^{(t)}\left(1-  \tau^{(t)}_i/\tau_{\max}\right)$ for each $i$, and $t \leftarrow t+1$.
}
\Return $\nu^{(t)}, w^{(t)}$.
\end{algorithm}

Our proof follows by tracing the argument of the soft down-weighting filter \cite{jerrynote2}. 
First,  we assume  that there exists a good set$G \subseteq [n]$ such that  $|G| \geq (1-\epsilon)n$ and 
\begin{equation}\label{eq:good-condtion}
        \frac{1}{(1-\epsilon)n}  \sum_{i\in G}  \left(x_{i}- \nu\right)\left(x_{i}- \nu\right)^{\top}\preceq \lambda I.
\end{equation}
Let $B = [n] \setminus G$, and we first establish a technical condition on $m^{(t)}$.
\begin{lemma}\label{lem:med-m}
Let $\beta^{(t)} = \tfrac{1}{n} \sum_{i\in G} \alpha_i^{(t)}$.  Then we have $|m^{(t)} - \beta^{(t)}|^2 \leq \frac{\lambda}{{1/2 -\epsilon}}$.
\end{lemma}
\begin{proof}
We fix one iteration and drop the superscript.
let $\mu_G = \frac{1}{n} \sum_{i\in G} x_i$.
First, observe that by \eqref{eq:good-condtion}, we have 
\begin{equation}\label{eq:good-condtion2}
        \frac{1}{(1-\epsilon)n}  \sum_{i\in G}  \left(x_{i}- \mu_G\right)\left(x_{i}- \mu_G\right)^{\top}\preceq \lambda I.
\end{equation}
Therefore, $\E_{i \sim G}\left [ (\alpha_i - \beta )^2\right]  = \E_{i\sim G} \langle v, \mu_G - x_i \rangle ^2\leq \lambda$.  By Chebyshev's inequality,
\begin{align}
    \Pr_{i\sim G} \left( |\alpha_i  - \beta| >\sqrt{  \frac{\lambda}{1/2-\epsilon}}  \right)\le \frac{1}{2}-\epsilon.
\end{align}
This means that we have $|G|\cdot (1/2+\epsilon)$ points $i\in [n]$ that  satisfy   $|\alpha_i  - \beta|^2 \leq  \frac{\lambda}{{1/2 -\epsilon}}$. Our claim now follows since $|G| \geq (1-\epsilon)$ and $(1-\epsilon)(1/2+\epsilon) > 1/2$ for any $\epsilon\in (0,1/2)$.
\end{proof}
This allows us to establish the key invariant of the algorithm.
\begin{lemma}\label{lem:monotone}
Suppose at iteration $s$, we have that 
\begin{equation}\label{eq:bigM}
    \left\| M^{(s)}\right\| \geq \frac{16}{7} \lambda \left( 1+ \frac{1}{1/2-\epsilon} \right).
\end{equation}
and for $t= s$ 
\begin{equation}\label{eq:key-inv}
    \sum_{i \in G} \frac{1}{n}-w_{i}^{(t)}<\sum_{i \in B} \frac{1}{n}-w_{i}^{(t)}
\end{equation}
Then the condition \eqref{eq:key-inv} continues to hold for $t=s+1$.
\end{lemma}
\begin{proof}
Observe that to prove the claim inductively, it suffices to show that for any $s$,
\begin{align}
    \sum_{i \in G} w_{i}^{(s)}-w_{i}^{(s+1)}<\sum_{i \in B} w_{i}^{(s)}-w_{i}^{(s+1)}.
\end{align}
We now just focus on these two iterations, drop the superscript and denote $w^{(s+1)}$ by $w'$.
By definition of the update step (line 7 of \autoref{alg:breakdown}), we just need to prove that 
\begin{align}\label{eq:key-invariant}
    \sum_{i \in G} w_{i} \tau_{i}<\sum_{i \in B} w_{i} \tau_{i}.
\end{align}
Now note that since $\nu = \mu(w) = \sum_{i=1}^n w_ix_i$, we have
\begin{align*}
    \sum_{i =1}^n w_{i} \tau_{i}  &= \sum_{i=1}^n w_i ( \langle v, x_i \rangle - m)^2
    \\& = \sum_{i=1}^n w_i \left( \langle v, x_i - \nu\rangle + \langle \nu, v\rangle - m\right)^2\\
    &=  \sum_{i=1}^n w_i \left( \langle v, x_i - \nu\rangle^2 +  (m-\langle v,\nu\rangle)^2\right)
    \\ 
    &\geq \sum_{i =1}^n w_{i}\left\langle v, x_{i}-\nu\right\rangle^{2}\\
    &=v^{\top} Mv\geq \frac{7}{8  }  \|M\|_{2}.
\end{align*}
Hence, to establish invariant \eqref{eq:key-invariant}, we proceed by showing that 
\begin{align}
    \sum_{i\in G} w_i \tau_i \leq \frac{7}{16}\|M\|_2. 
\end{align}
Since $w_i \leq \frac{1}{n}$, we have $\sum_{i\in G} w_i \tau_i \leq \sum_{i\in G} \tfrac{1}{n} ( \langle v, x_i \rangle - m)^2 $. On the other hand, let $\mu_G = \frac{1}{n} \sum_{i\in G} x_i$, and so by condition~\eqref{eq:good-condtion} and \autoref{lem:med-m},
\begin{align*}
    \sum_{i\in G} \frac{1}{n} \left( \langle v, x_i \rangle - m\right)^2 &= \frac{1}{n}\sum_{i \in G}\left\langle v, x_{i}-\mu_{G}\right\rangle^{2}+\left|\langle\mu_{G},v\rangle-m\right|^{2}\\
    &\le \lambda + \frac{\lambda}{{1/2 -\epsilon}}\\
    &\leq \frac{7}{16} \|M\|,
\end{align*}
by our assumption~\eqref{eq:bigM}. This completes the proof. 
\end{proof}

\begin{theorem}
For any $\epsilon \in (0,1/2)$, \autoref{alg:breakdown} gives a constant approximation to the spectral sample reweighting problem (\autoref{def:ssr}). The algorithm terminates in $T = O(n)$ iterations.
\end{theorem}
\begin{proof}
The run-time follows from the invariant \autoref{lem:monotone}, which guarantees weights on bad points are removed more than those on good points. Hence, after $2\epsilon n$ iterations, the algorithm must terminate. 
Moreover, when the algorithm terminates, we   have 
\begin{equation}
    \left\| M^{(t)}\right\| \le  \frac{16}{7} \lambda \left( 1+ \frac{1}{1/2-\epsilon} \right).
\end{equation}
For any constant $\epsilon \leq 1/2 - O(1)$, the bound is $O(\lambda)$. 
\end{proof}
\paragraph{Robust mean estimation.}
Our reduction from spectral sample reweighting to robust mean estimation is not sufficiently tight for the purpose of attaining optimal breakdown point.
 Instead, we need to appeal to the following more refined spectral signature.
\begin{claim}[refined spectral signature~\cite{jerrynote}]\label{clm:refine-ss}
Let $S = S_g \cup S_b\setminus S_r$  be $n$ points with $|S_b| = |S_r| = \epsilon n$. Define $\mu_g = \frac{1}{n} \sum_{i\in S_g} x_i$ and $\Sigma =  \frac{1}{n} \sum_{i\in S_g} (x_i-\mu)(x_i-\mu)^\top$.  Let $w(S)$  be the uniform distribution on $S$ and $\mathcal C_{n,\epsilon} = \{ w : \|w - w(S)\|_1 \leq \epsilon,  0 \leq w_i \leq 1/n \text{ for }  i \in [n] \}$. Then for any $w\in \mathcal{C}_{n,\epsilon}$,
\begin{align*}
\left(\sum_{i\in S\cap S_g} w_i\right)    \left\|\mu - \mu(w)\right\| \leq  \sqrt{2\epsilon \|\Sigma \|} +  \sqrt{\epsilon    \left\| \Sigma(w)\right\| } .
\end{align*}
\end{claim}

\begin{theorem}
For the problem of robust mean estimation (under bounded second moment), \autoref{alg:breakdown}   attains the optimal estimation error $O(\sqrt{\epsilon})$ for any $\epsilon<1/2$.
\end{theorem}
\begin{proof}
By \autoref{lem:monotone}, our algorithm always removes more weights from bad points than from good points.  Thus,  $w^{(t)} \in \mathcal{C}_{n,2\epsilon}$, as there are at most $\epsilon n$ bad points. Moreover,  $\sum_{i\in S\cap S_g} w_i \geq 1-2\epsilon$. 

For robust mean estimation, if we have $n = \Omega (d\log d / \epsilon)$ samples, then $\|\mu_g - \mu\| \leq O(\sqrt{\epsilon})$ and $\lambda  =\|\Sigma\| \leq 2$~\cite{diakonikolas2019robust}. Hence, applying \autoref{clm:refine-ss} and the guarantee that $\|M^{(t)}\| \leq  \frac{16}{7} \lambda \left( 1+ \frac{1}{1/2-\epsilon} \right)$,
\begin{align*}
    \left\|\mu_g - \nu^{(t)}\right\| \leq  \frac{1}{1-2\epsilon} \left(2\sqrt{\epsilon \|\Sigma\|} +  \sqrt{2\epsilon    \left\| M^{(t)}\right\| }\right) \leq O(\sqrt{\epsilon}),
\end{align*}
for any $\epsilon < 1/2$.  Finally, triangle inequality implies that $\|\mu -  \nu^{(t)}\| = O(\sqrt{\epsilon})$.
 \end{proof}

\end{document}